\documentclass[aps,pra,reprint,superscriptaddress]{revtex4-1}

\usepackage{adjustbox}
\usepackage{phaistos}
\usepackage{hyperref}      
\usepackage{url}           
\usepackage{booktabs}      
\usepackage{amsfonts}      
\usepackage{microtype}      
\usepackage{natbib}
\usepackage[normalem]{ulem}
\usepackage{bm}
\usepackage{braket}
\usepackage{graphicx}
\usepackage{subfigure} 
\usepackage{amsthm}
\usepackage{algorithmic}
\usepackage[linesnumbered,ruled]{algorithm2e}
\SetKwInOut{Parameter}{parameter}
\newtheorem{thm}{Theorem}
\newtheorem{lem}{Lemma}
\newtheorem{coro}{Corollay}
\newtheorem{definition}{Definition}

\newtheorem{assumption}{Assumption}
\usepackage{mathtools}
 
\usepackage{caption} 
\usepackage[table,xcdraw]{xcolor}

\DeclareMathOperator{\btheta}{\bm{\theta}}
\DeclareMathOperator{\Tr}{Tr}

\DeclareMathOperator{\RY}{R_Y}
\DeclareMathOperator{\CRY}{CR_Y}
\DeclareMathOperator{\RZ}{R_Z}

\DeclareMathOperator{\MMD}{MMD}
\newcommand{\Pro}{\mathbb{P}}

\DeclareMathOperator{\PP}{\mathbb{P}}
\newcommand{\rx}{\bm{x}}
\newcommand{\ry}{\bm{y}}
\DeclareMathOperator{\QQ}{\mathbb{Q}}
\DeclareMathOperator{\EE}{\mathbb{E}}
\DeclareMathOperator{\GHZ}{\text{GHZ}} 

\DeclareMathOperator{\Ccal}{\mathcal{C}}

\DeclareMathOperator{\Hcal}{\mathbb{H}}

\newcommand{\pbra}[1]{\left( #1 \right)}

\begin{document}

\title{Power of Quantum Generative Learning}
 
\author{Yuxuan Du}
\affiliation{JD Explore Academy, Beijing 101111, China}

\author{Zhuozhuo Tu}
\affiliation{School of Computer Science, The University of Sydney, Darlington, NSW 2008, Australia}

\author{Bujiao Wu}
\affiliation{Center on Frontiers of Computing Studies, Peking University, Beijing 100871, China}
\affiliation{School of Computer Science, Peking University, Beijing 100871, China}

\author{Xiao Yuan}
\affiliation{Center on Frontiers of Computing Studies, Peking University, Beijing 100871, China}
\affiliation{School of Computer Science, Peking University, Beijing 100871, China}

\author{Dacheng Tao}
\affiliation{JD Explore Academy, Beijing 101111, China}

\begin{abstract}
	The intrinsic probabilistic nature of quantum mechanics invokes endeavors of designing quantum generative learning models (QGLMs).	Despite the empirical achievements, the foundations and the potential advantages of QGLMs remain largely obscure. To narrow this knowledge gap, here we explore the generalization property of QGLMs, the capability to extend the model from learned to unknown data.  We consider two prototypical QGLMs, quantum circuit Born machines and quantum generative adversarial networks, and explicitly give their generalization bounds. The result identifies superiorities of QGLMs over classical methods when quantum devices can directly access the target distribution and quantum kernels are employed. We further employ these generalization bounds to exhibit potential advantages in quantum state preparation and Hamiltonian learning. Numerical results of QGLMs in loading Gaussian distribution and estimating ground states of parameterized Hamiltonians accord with the theoretical analysis.  Our work opens the avenue for quantitatively understanding the power of quantum generative learning models.
\end{abstract}

\maketitle
 
\section{Introduction}
Learning is a generative activity that constructs its own interpretations of information and draws inferences on them \cite{wittrock1989generative}. This comprehensive philosophy sculpts a substantial subject in  artificial intelligence, which is designing powerful generative learning models (GLMs) \cite{lecun2015deep,schmidhuber2015deep} to capture the distribution $\mathbb{Q}$ describing the real-world data (see Fig.~\ref{fig:QCBM_scheme}(a)). Concisely, a fundamental concept behind GLMs is estimating $\mathbb{Q}$ by a tunable probability distribution $\mathbb{P}_{\btheta}$. In the past decades, a plethora of GLMs, e.g., the Helmholtz machine  \cite{dayan1995helmholtz}, variational auto-encoders \cite{Kingma2013VAE}, and generative adversarial networks (GANs) \cite{goodfellow2014generative,zhao2017energy},  have been proposed. Attributed to the efficacy and flexibility of handling  $\mathbb{P}_{\btheta}$, these GLMs have been broadly applied to myriad scientific domains and gained remarkable success, including image synthesis and editing \cite{YiLLR19,karras2019style,park2019semantic}, medical imaging \cite{armanious2020medgan}, molecule optimization \cite{maziarka2020mol,mendez2020novo}, and quantum computing \cite{ahmed2021quantum,carrasquilla2019reconstructing,smith2021efficient,rocchetto2018learning,melko2019restricted,carleo2018constructing}. Despite the  wide success, their limitations have recently been recognized from different perspectives, such as expensive runtime and inferior performance towards complex datasets \cite{Du2019implicit,dosovitskiy2016generating,arora2018do,che2016mode,mescheder2018training,gui2020review}.

Envisioned by the intrinsic probabilistic nature of quantum mechanics and the superior power of quantum computers \cite{feynman2017quantum,biamonte2017quantum,harrow2017quantum},    quantum generative learning models (QGLMs) are widely believed to enhance the ability of GLMs. Concrete evidence has been provided by Refs. \cite{gao2018quantum,lloyd2018quantum}, showing that fault-tolerant based QGLMs could surpass GLMs with  stronger model expressivity and runtime speedups. Since fault-tolerant quantum computing is still in absence, attention has recently shifted to design QGLMs that can be efficiently carried out on noisy intermediate-scale quantum (NISQ) machines \cite{preskill2018quantum,arute2019quantum,wu2021strong} with  advantages on specific tasks \cite{huang2021quantum,huang2021power,Wang2021towards,du2021exploring}. For this purpose, a leading strategy is constructing QGLMs through variational quantum algorithms (VQAs) \cite{cerezo2021variational,bharti2021noisy}. Depending on  \textit{whether the probability distribution $\PP_{\btheta}$ is explicitly formulated or not} \footnote{Implicit probabilistic models do not specify the distribution of the data itself, but rather define a stochastic process that, after training, aims to draw samples from the underlying data distribution. Refer to SM \ref{apped:sum_GLMs} for details}, these QGLMs can be mainly divided into explicit QGLMs and implicit QGLMs. Primary protocols of these two classes are  \textit{quantum circuit Born machines} (QCBMs) \cite{benedetti2019generative,liu2018differentiable,coyle2021quantum,du2018expressive} and \textit{quantum generative adversarial networks} (QGANs) \cite{zoufal2019quantum,zeng2019learning,huang2021experimental,huang2021quantum2,romero2021variational,gili2022evaluating,nakaji2021quantum,braccia2022quantum,anand2021noise,du2022Effcient,huang2021quantum2}.  Experimental studies have demonstrated the feasibility of QGLMs for different learning tasks, e.g., image generation \cite{zhu2019training,huang2021experimental}, state approximation \cite{huang2021quantum,hu2019quantum}, and drug design \cite{li2021quantum,Jin2022Quantum}.  

A crucial vein in quantum machine learning \cite{biamonte2017quantum} is comprehending potential advantages of quantum learners from the perspective of \textit{generalization}, which   quantifies the capability to extend the model from learned to unknown data \cite{arora2017generalization,allen2019learning,du2021learnability,zhao2018bias}. In sharp contrast to quantum discriminative learning    \cite{abbas2020power,banchi2021generalization,caro2021generalization,du2021efficient,huang2021information,Junyu2022dynamic,qian2021dilemma}, prior literature related to the generalization  of QGLMs is \textit{scarce and negative} \cite{sweke2021quantum,hinsche2021learnability}. Concretely, Ref.~\cite{hinsche2021learnability} pointed out the hardness for both QCBM and GLMs to efficiently learn the output distributions of local quantum circuits under the statistical query access. In this regard,  two key questions are naturally invoked: how to quantify the generalization of different QGLMs in a generic way? And is there any potential advantage of QGLMs in solving nontrivial learning problems?

\begin{figure*} 
\centering
\includegraphics[width=0.92\textwidth]{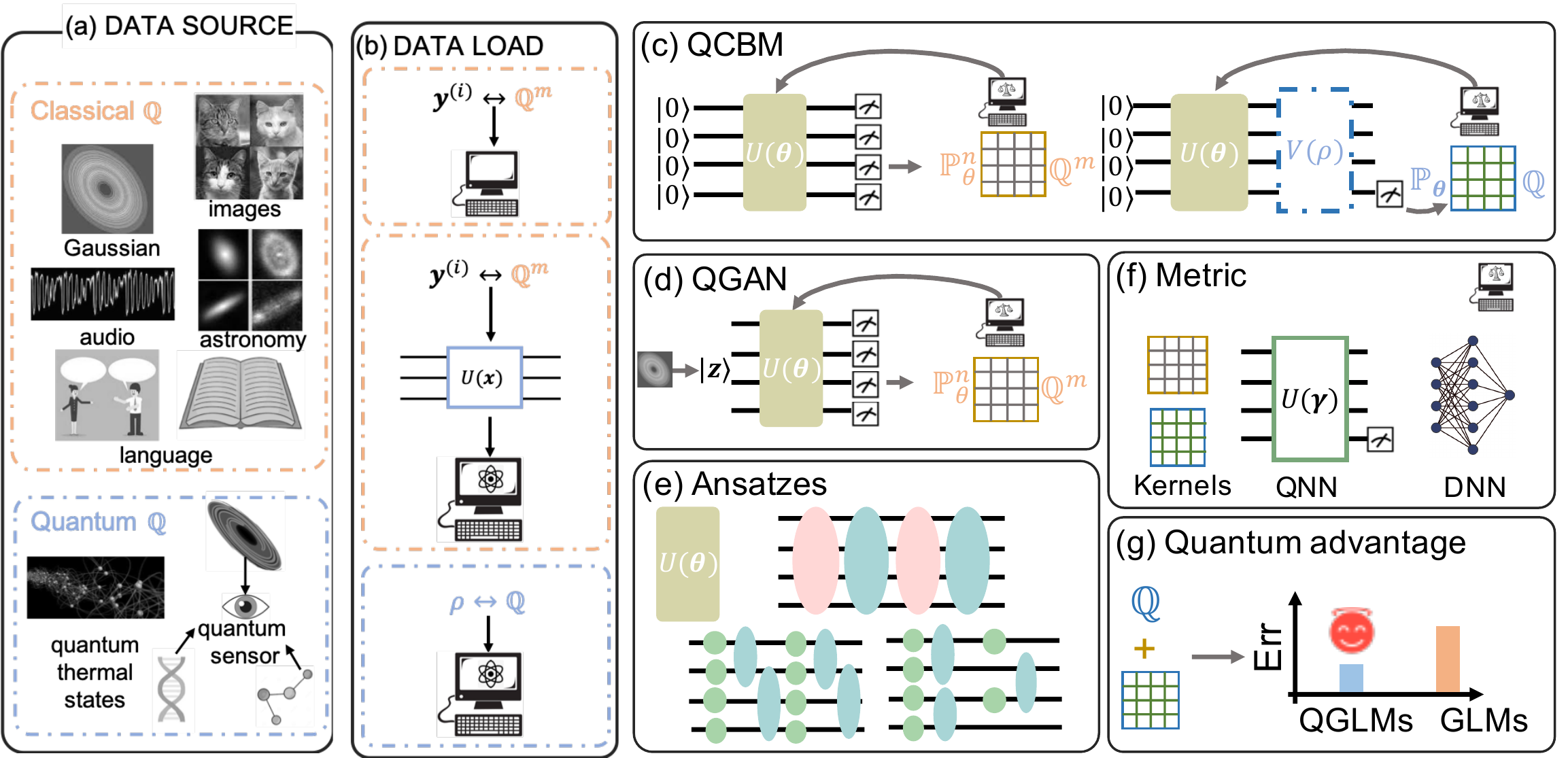}
\caption{\small{\textbf{The paradigm of quantum generative learning models}. (a) The data  explored in  generative learning  includes both classical and quantum scenarios.  (b) The approaches of QGLMs to access target data distribution $\QQ$. When $\QQ$ is classical, QGLMs operate its samples on the classical side or encode its samples into quantum circuits. When $\QQ$ is quantum, QGLMs may directly access it without sampling.  (c) The paradigm of QCBMs with MMD loss. The left and right panels depict the setup of classical and quantum kernels, respectively. (d) The scheme of QGANs with MMD loss for the continuous distribution $\QQ$. (e) QAOA, hardware-efficient, and tensor-network based Ans\"atze are covered by $\hat{U}(\btheta)$ in Eq.~(\ref{eqn:ansatz}). (f) The metrics exploited in QGLMs to measure the discrepancy between the generated and target distributions. (g) When $\QQ$ and $k(\cdot,\cdot)$ are both quantum, QGLMs may attain generalization advantages over classical GLMs.}}	
\label{fig:QCBM_scheme}
\end{figure*}

To shrink the above knowledge gap, here we establish the generalization theory of QGLMs with the maximum mean discrepancy (MMD)  loss \cite{Gretton2012AKernel} and utilize these results to affirm  potential merits of QGLMs in practical applications. The attention on MMD loss is because it is an efficient measure to evaluate the difference between  $\PP_{\btheta}$ and $\QQ$ without the curse of dimensionality issue. Our first result is unveiling the power of QCBMs and QGANs through the lens of the statistical learning theory \cite{vapnik1999nature}. Specifically, when $\QQ$ is discrete, we prove that the  generalization error of QCBMs  scales with $\tilde{O}(\sqrt{1/n + 1/m})$, where $n$ and $m$ refer to the number of examples sampled from $\PP_{\btheta}$ and $\QQ$. In addition, when $\QQ$ can be efficiently prepared by quantum circuits (see Fig.~\ref{fig:QCBM_scheme}(b)), quantum kernels enable QCBMs to attain a strictly lower generalization error than  classical kernels. When $\QQ$ is continuous, we prove that the generalization error of QGAN is upper bounded by  $\tilde{O}(\sqrt{1/n + 1/m} + Nd^{k}\sqrt{N_{gt}N_{ge}}/n)$, where   $d$, $k$, $N_{ge}$, and $N_{gt}$ refer to the dimension of a qudit, the type of quantum gates, the number of encoding gates, and the number of trainable parameters, respectively. This bound explicitly reveals how the encoding strategy and the adopted Ansatz affect the generalization. Taken together, sufficient expressivity and large examples allows $\PP_{\btheta} \approx \QQ$ for QGLMs. In light of this observation, we exhibit the potential advantage of QCBMs and QGANs in  quantum state preparation and parameterized Hamiltonian learning. 
    
The remainder of this study is organized as follows. In Secs.~\ref{sec:gene-QCBM} \& \ref{sec:gene-QGAN}, we theoretically explore the power of QCBM and QGAN through the lens of statistical learning theory, respectively. Then, in Sec.~\ref{sec:application-QGLM}, we elaborate how QGLMs can advance quantum state preparation and Hamiltonian learning. Subsequently, we conduct numerical simulations to validate our claims about QCBMs and QGANs in Sec.~\ref{sec:numerical}.  We conclude this study in Sec.~\ref{sec:discuss}. 
 
\section{Generalization of QCBM}\label{sec:gene-QCBM}
We first study the generalization property of QCBM. As shown in Fig.~\ref{fig:QCBM_scheme}(a), QCBM applies an $N$-qudit Ansatz $\hat{U}(\bm{\theta})$ to a fixed input state $\rho_0 = (\ket{0}\bra{0})^{\otimes N}$ to form the parameterized distribution $\PP_{\btheta}\in \PP_{\bm{\Theta}}$, where $\btheta$ denotes trainable parameters living in the parameter space $\bm{\Theta}$. The probability of sampling $i \in[d^N]$ over the \textit{discrete} distribution $\PP_{\btheta}$ yields
\begin{equation}\label{eqn:QCBM}
\PP_{\btheta}(i)= \Tr(\Pi_i \hat{U}(\bm{\theta})\rho_0 \hat{U}(\bm{\theta})^\dagger),
\end{equation}
where $\Pi_i=\ket{i}\bra{i}$ refers to the projector of the computational basis $i$. Given $n$ reference examples $\{\rx^{(j)}\}_{j=1}^n$ sampled from $\PP_{\btheta}$, its  empirical distribution is defined as $\PP_{\btheta}^n(i)=  \sum_{j=1}^n\delta_{\rx^{(j)}}(i)/n$ with $\delta_{(\cdot)}(\cdot)$ being the indicator. Throughout the whole study, the Ansatz employed in QCBMs takes the generic form
\begin{equation}\label{eqn:ansatz}
	\hat{U}(\bm{\theta})=\prod_{l=1}^{N_g}\hat{U}_l(\bm{\theta}),
\end{equation}
where   $\hat{U}_l(\bm{\theta})\in\mathcal{U}(d^k)$ refers to the $l$-th quantum gate operated with at most $k$-qudits with $k\leq N$, and $\mathcal{U}(d^k)$ denotes the unitary group in dimension $d^k$ ($d=2$ for qubits). Note that for  simplicity, the definition of $N$-qudit Ansatz $\hat{U}(\bm{\theta})$ in Eq.~(\ref{eqn:ansatz}) omits some identity operators. The complete description for the $l$-th layer is $   \mathbb{I}_{d^{N-k}}\otimes \hat{U}_l(\bm{\theta})$. The form of $\hat{U}(\btheta)$  covers almost all Ans\"atze  in VQAs and some  constructions are given in Fig.~\ref{fig:QCBM_scheme}(e).

The training of QCBMs is guided by the maximum mean discrepancy (MMD) loss \cite{Gretton2012AKernel}. Suppose  $\QQ$ and $\PP_{\bm{\theta}}$ live on the same space. The MMD loss measures the difference between $\PP_{\btheta}$ and $\QQ$ with  
\[
	\MMD^2(\PP_{\btheta}||\QQ) = \EE(k(\rx,\rx')) + \EE(k(\ry,\ry'))- 2\EE(k(\rx,\ry)),\] where the expectations are taken over the randomness of examples and $k(\cdot,\cdot)$ denotes a predefined kernel.  QCBMs aim to find an estimator minimizing  MMD loss,  
\begin{equation}\label{eqn:opt_MMD}
	\hat{\btheta} = \arg\min_{\btheta\in \bm{\Theta}} \MMD^2(\PP_{\btheta}||\QQ). 
\end{equation}
If the distribution $\PP_{\btheta}$ and $\QQ$ cannot be accessed directly, the estimator is located by minimizing the empirical MMD loss with finite samples \cite{Gretton2012AKernel}, i.e., 
\begin{equation}\label{eqn:opt_MMD_empirical}
	\hat{\btheta}^{(n,m)} = \arg\min_{\btheta\in \bm{\Theta}} \MMD_U^2(\PP_{\btheta}^n||\QQ^m),
\end{equation}
where $\PP_{\btheta}^n$ and $\QQ^m$ separately refers to the empirical distribution of $\PP_{\btheta}$ and $\QQ$, and 
\begin{align}
	& \MMD^2_U(\PP_{\btheta}^n||\QQ^m) :=\frac{1}{n(n-1)}\sum_{i\neq i'}^n k(\rx^{(i)}, \rx^{(i')}) 
	\nonumber \\
	& + \frac{1}{m(m-1)}\sum_{j\neq j'}^m k(\ry^{(j)}, \ry^{(j')})- \frac{2}{nm}\sum_{i,j} k(\rx^{(i)}, \ry^{(j)}). \nonumber 
\end{align}
More explanations of MMD loss are provided in SM \ref{appendix:MMD_opt}. 

The choice of the kernel $k(\cdot, \cdot)$ in MMD loss is flexible.   When it is specified to be the \textit{quantum} kernel (i.e., the specified kernel can be efficiently computed by quantum algorithms, where  representative examples are linear kernel and polynomial kernels \cite{schuld2019quantum}) and  $\QQ$ can be directly accessed by QGLMs (e.g.,  $\QQ$ can be efficiently prepared by a  quantum circuit), the MMD loss can be efficiently calculated. 
\begin{lem}\label{lem:equi-QCBM-QGAN}
Suppose that the distribution $\QQ$ can be directly accessed by QCBMs. When the quantum kernel is adopted, the MMD loss in Eq.~(\ref{eqn:opt_MMD}) can be  estimated within an error $\epsilon$ in $O({1}/{\epsilon}^2)$ runtime complexity.
\end{lem}
\begin{proof}[Proof of Lemma \ref{lem:equi-QCBM-QGAN}]
Here we separately elaborate on the calculation of MMD loss when the target distribution $\QQ$ can be efficiently prepared by a diagonalized mixed state $\sigma =\sum_{\ry} \mathbb{Q}(\ry) \ket{\ry}\bra{\ry}$ or a pure quantum state $\ket{\Psi} =\sum_{\ry} \sqrt{\mathbb{Q}(\ry)}\ket{\ry}$.

\medskip
\noindent\textit{\underline{Diagonalized mixed states}}. In this setting, the quantum kernel corresponds to the linear kernel, i.e., \[k(\bm{a},\bm{a}')=\langle \bm{a}, \bm{a}' \rangle.\]  Recall the definition of the MMD loss is $\MMD^2(\PP_{\btheta}||\QQ) = \EE(k(\rx,\rx')) - 2\EE(k(\rx,\ry)) + \EE(k(\ry,\ry'))$. The first term equals to 
\begin{equation}
	\EE\left(k(\rx,\rx')\right)  = \EE_{\PP_{\btheta} , \PP_{\btheta}}\left( \delta_{\rx,\rx'} \right) = \sum_{\rx} \PP_{\btheta}^2(\rx). \nonumber
\end{equation}
Similarly, the second term equals to
\begin{equation}
	\EE\left(k(\rx,\ry)\right) = \sum_{\rx} \PP_{\btheta}(\rx)\QQ(\rx). \nonumber
\end{equation}
And the third term equals to
\begin{equation}
	\EE\left(k(\ry,\ry')\right) = \sum_{\ry} \QQ^2(\ry). \nonumber
\end{equation}
The above three terms can be effectively and analytically evaluated by quantum Swap test when the input state of QCBM in Eq.~(\ref{eqn:QCBM}) is a full rank mixed state, e.g., $\rho_0=\mathbb{I}_{2^N}/2^N$. Denote the output state of QCBM   as $\rho=U(\btheta)\rho_0U(\btheta)^{\dagger}$. This state is also diagonalized and its diagonalized entry records $\PP_{\btheta}$, i.e.,
\begin{equation}
	\rho= \sum_{\rx}\PP_{\btheta}(\rx)\ket{\rx}\bra{\rx}. \nonumber
\end{equation} 
According to \cite{buhrman2001quantum, kobayashi2003quantum}, given two mixed states $\varrho_1$ and $\varrho_2$, the output of Swap test is $1/2+\Tr(\varrho_1\varrho_2)/2$ with an additive error $\epsilon$ in $O(1/\epsilon^2)$ runtime. As such, when $\varrho_1=\varrho_2=\rho$, the first term $\EE\left(k(\rx,\rx')\right)$ can be calculated by Swap test, because  $\Tr(\rho\rho)=\sum_{\rx}\PP_{\btheta}^2(\rx)$. Likewise, through setting $\varrho_1=\rho$, and $\varrho_2=\sigma$ ($\varrho_1=\varrho_2=\sigma$), the second (third) term can be efficiently evaluated by Swap test with an additive error $\epsilon$. In other words, by leveraging Swap test, we can estimate MMD loss with an additive error $\epsilon$ in $O(1/\epsilon^2)$ runtime cost.

\medskip
\noindent\textit{\underline{Pure states}}.  The quantum kernel in this scenario corresponds to a nonlinear kernel, i.e., \[k(\bm{a},\bm{a}')=\langle \frac{\bm{a}}{\sqrt{\PP(\bm{a})}}, \frac{\bm{a}'}{\sqrt{\PP(\bm{a}')}} \rangle,\] where $\PP(\bm{a})$ stands for the probability of sampling $\bm{a}$ and $\sum_{\bm{a}}\PP(\bm{a})=1$. With this regard, the explicit form of MMD loss yields 
\begin{eqnarray}
	&& \MMD(\PP_{\btheta}||\QQ) \nonumber\\\
	= && \EE(k(\rx,\rx')) - 2\EE(k(\rx,\ry)) + \EE(k(\ry,\ry')) \nonumber\\
	= && \sum_{\rx}\sum_{\rx'} \PP_{\btheta}\PP_{\btheta} \left\langle \frac{\rx}{\sqrt{\PP_{\btheta}(\rx)}}, \frac{\rx'}{\sqrt{\PP_{\btheta}(\rx')}} \right \rangle \nonumber\\
	&& -2  \sum_{\rx}\sum_{\ry} \PP_{\btheta}\QQ \left\langle \frac{\rx}{\sqrt{\PP_{\btheta}(\rx)}}, \frac{\ry}{\sqrt{\QQ(\ry)}} \right \rangle \nonumber\\
	&& +  \sum_{\ry}\sum_{\ry'} \QQ\QQ \left\langle \frac{\ry}{\sqrt{\QQ(\ry)}}, \frac{\ry'}{\sqrt{\QQ(\ry')}} \right \rangle \nonumber\\
	= && \sum_{\rx} \PP_{\btheta}(\rx) +  \sum_{\ry} \QQ(\ry) - 2\sum_{\rx} \sqrt{\PP_{\btheta}(\rx)\QQ(\rx)} \nonumber\\
	= && 2 - 2\sum_{\rx}\sqrt{\PP_{\btheta}(\rx)\QQ(\rx)}. \nonumber
\end{eqnarray}

The above results indicate that the evaluation of $\MMD$ loss amounts to calculating $\sum_{\rx}\sqrt{\PP_{\btheta}(\rx)\QQ(\rx)}$. Denote the generated state of QCBM as $\ket{\Phi(\btheta)}=U(\btheta)\ket{0}^{\otimes n}=e^{i\phi}\sum_{\rx}\PP_{\btheta}(\rx)\ket{\rx}$, where $\rho_0=(\ket{0}\bra{0})^{\otimes n}$. When the target distribution $\QQ$ refers to a pure quantum state $\ket{\Psi}=\sum_{\ry}\sqrt{\QQ(\ry)}\ket{\ry}$, the term $\sum_{\rx}\sqrt{\PP_{\btheta}(\rx)\QQ(\rx)}$ can be evaluated by Swap test \cite{buhrman2001quantum}, i.e.,
\begin{equation}
|\braket{\Phi(\btheta)|\Psi}|^2   = \left(\sum_{i=1}^{2^N}  \PP_{\btheta}(\rx) \QQ(\rx) \right)^2.
\end{equation}
According to \cite{buhrman2001quantum}, taking account of the sample error, this term can be estimated within an additive error $\epsilon$ in $O(1/\epsilon^2)$ runtime complexity.  
\end{proof}
\noindent It hints that when both $k(\cdot,\cdot)$ and $\QQ$ are quantum, $\MMD^2(\PP_{\btheta}||\QQ)$ can be efficiently calculated in which the runtime cost is \textit{independent} of the dimension of data space. In  contrast, for GLMs and QCBMs with classical kernels, the runtime cost in computing the MMD loss polynomially scales with $n$ and $m$. Such runtime discrepancy warrants the advantage of QCBMs explained in the subsequent context.

We now analyze the generalization ability of QCBMs. To begin with, let us extend the definition of \textit{generalization error} from the classical learning theory \cite{dziugaite2015training}  to the regime of quantum generative learning. 
\begin{definition}[Generalization error of QGLMs]\label{def:generalization}
	When either the kernel $k(\cdot,\cdot)$ or the target $\QQ$ is \textit{classical}, the generalization error of QGLMs is
\begin{equation}\label{eqn:gene_error_bound-emp}
	\mathfrak{R}^C = \MMD^2 (\PP_{\hat{\btheta}^{(n,m)}}||\QQ) - \inf_{\btheta\in\bm{\Theta}}\MMD^2 (\PP_{\bm{\theta}}||\QQ).
\end{equation}
When the kernel $k(\cdot,\cdot)$ is \textit{quantum} and  $\QQ$ can be efficiently accessed by quantum machines, the generalization error of QGLMs is 
\begin{equation}\label{eqn:gene_error_bound-exac}
	\mathfrak{R}^Q = \MMD^2 (\PP_{\hat{\btheta}}||\QQ) - \inf_{\btheta\in\bm{\Theta}}\MMD^2 (\PP_{\bm{\theta}}||\QQ).
\end{equation}
\end{definition}
\noindent Intuitively, both $\mathfrak{R}^C$ and $\mathfrak{R}^Q$  evaluate the divergence of the estimated and the optimal MMD loss, where a lower $\mathfrak{R}^C$ or $\mathfrak{R}^Q$ suggests a better generalization ability. The following theorem establishes the generalization theory of QCBMs.

\begin{thm}\label{thm:Gene-QCBM}
Assume  $ \max_{\btheta\in\bm{\Theta}}\MMD^2(\PP_{\btheta}||\QQ)\leq C_1$ with $C_1$ being a constant. Define $C_2=\sup_{\rx} k(\rx, \rx)$. Following settings in Lemma \ref{lem:equi-QCBM-QGAN}, with probability at least $1-\delta$, the generalization error of QCBMs yields 
	\begin{equation}\label{eqn:thm1-1}
	 \mathfrak{R}^Q\leq \mathfrak{R}^C \leq  C_1\sqrt{\frac{8}{n}+\frac{8}{m}}\sqrt{C_2}(2+\sqrt{\log \frac{1}{\delta} }).
	\end{equation}  
\end{thm}
\begin{proof}[Proof of Theorem \ref{thm:Gene-QCBM}] 

To prove Theorem \ref{thm:Gene-QCBM}, we first derive the upper bound of the generalization error $\mathfrak{R}^C$ under the generic setting. Then, we analyze $\mathfrak{R}^Q$ under the specific setting where the quantum kernel is employed and the target distribution $\QQ$ can be directly accessed by quantum machines. The analysis of $\mathfrak{R}^C$  adopts the following lemma.
\begin{lem}[Adapted from Theorem 1, \cite{briol2019statistical}] \label{lem:lem2-dist-mmd}
Suppose that the kernel $k(\cdot,\cdot)$ is bounded. Following the notations in Theorem \ref{thm:Gene-QCBM}, when the number of examples sampled from $\PP_{\hat{\btheta}^{(n,m)}}$ and $\QQ$ is $n$ and $m$,  with probability $1-\delta$, $
	\MMD(\PP_{\hat{\btheta}^{(n,m)}} ||\QQ)\leq \inf_{\btheta\in\bm{\Theta}} \MMD(\PP_{\btheta}||\QQ) + 2\sqrt{\frac{2}{n}+\frac{2}{m}}\sqrt{\sup_{\bm{x}\in\mathcal{X}}k(\rx, \rx)}(2+\sqrt{\frac{2}{\delta}})$. 
\end{lem}

\noindent\textit{The calculation of $\mathfrak{R}^C$.} Recall the definition of $\hat{\btheta}^{(n,m)}$ in Eq.~(\ref{eqn:opt_MMD_empirical}). Let us first rewrite  the generalization error as   
\begin{eqnarray}\label{eqn:append-thm1-0}
	\mathfrak{R}^C = &&  \MMD^2 (\PP_{\hat{\btheta}^{(n,m)}}||\QQ) - \inf_{\btheta\in\bm{\Theta}}\MMD^2 (\PP_{\btheta}||\QQ) \nonumber \\
	= && \left(\MMD (\PP_{\hat{\btheta}^{(n,m)}}||\QQ) - \inf_{\btheta\in\bm{\Theta}}\MMD (\PP_{\btheta}||\QQ) \right) \nonumber\\
	&& \times \left(\MMD (\PP_{\hat{\btheta}^{(n,m)}}||\QQ) + \inf_{\btheta\in\bm{\Theta}}\MMD (\PP_{\btheta}||\QQ) \right)\nonumber\\
	\leq && 2C_1\left|\MMD (\PP_{\hat{\btheta}^{(n,m)}}||\QQ) - \inf_{\btheta\in\bm{\Theta}}\MMD (\PP_{\btheta}||\QQ) \right|,\nonumber
\end{eqnarray}
where the second equality uses  $\inf_{\btheta\in\bm{\Theta}}\MMD^2 (\PP_{\btheta}||\QQ) = (\inf_{\btheta\in\bm{\Theta}}\MMD (\PP_{\btheta}||\QQ))^2 $  and the inequality employs the $\MMD (\PP_{\hat{\btheta}^{(n,m)}}||\QQ)+\inf_{\btheta\in\bm{\Theta}}\MMD (\PP_{\btheta}||\QQ) \leq 2\MMD (\PP_{\hat{\btheta}^{(n,m)}}||\QQ)\leq 2C_1$.

In conjunction with the above equation with the results of Lemma \ref{lem:lem2-dist-mmd}, we obtain that with probability at least $1-\delta$, the upper bound of the generalization error of QGCM yields 
\begin{equation}\label{eqn:append-thm1-1}
	\mathfrak{R}^C \leq  4C_1\left(\frac{2}{n}+\sqrt{\frac{2}{m}}\right)\sqrt{\sup_{x\in \mathcal{X}}k(\rx, \rx)}\left(2+\sqrt{\log \frac{2}{\delta}}\right).
\end{equation}

\noindent\textit{The calculation of $\mathfrak{R}^Q$.} Recall the definition of $\hat{\btheta}$ in Eq.~(\ref{eqn:opt_MMD}). When QCBM adopts the quantum kernel and the target distribution $\QQ$ can be directly accessed by quantum machines, the minimum argument of the loss function yields $\hat{\btheta} = \arg\min_{\btheta\in \bm{\Theta}}\MMD^2(\PP_{\btheta}||\QQ)$. Following the definition of the generalization error, we obtain
\begin{eqnarray}\label{eqn:append-thm1-2}
	\mathfrak{R}^Q = &&  \MMD^2 (\PP_{\hat{\btheta}}||\QQ) - \inf_{\btheta\in\bm{\Theta}}\MMD^2 (\PP_{\btheta}||\QQ) \nonumber\\
	&&\leq \MMD^2 (\PP_{\hat{\btheta}^{(n,m)}}||\QQ) - \inf_{\btheta\in\bm{\Theta}}\MMD^2 (\PP_{\btheta}||\QQ) \nonumber\\
	&&= \mathfrak{R}^C,
\end{eqnarray}
where the inequality is supported by the definition of $\hat{\btheta}$, i.e,  $\MMD^2 (\PP_{\hat{\btheta}}||\QQ)=\min_{\btheta\in\bm{\Theta}}\MMD^2 (\PP_{\btheta}||\QQ)   \leq \MMD^2 (\PP_{\hat{\btheta}^{(n,m)}}||\QQ)$. 

Combining the results of $\mathfrak{R}^C$ and $\mathfrak{R}^Q$ in Eqs.~(\ref{eqn:append-thm1-1}) and (\ref{eqn:append-thm1-2}), we obtain that with probability at least $1-\delta$,
\begin{equation}
  \mathfrak{R}^Q \leq \mathfrak{R}^C \leq 4C_1\left(\frac{2}{n}+\sqrt{\frac{2}{m}}\right)\sqrt{\sup_{x\in \mathcal{X}}k(\rx, \rx)}\left(2+\sqrt{\log \frac{2}{\delta}}\right). \nonumber
\end{equation}  
\end{proof}
\noindent It indicates that when  $\QQ$ is quantum, quantum kernels promise a \textit{strictly lower} generalization error than classical kernels. Moreover, the decisive factor to improve generalization of QCBMs is to simultaneously increase $n$ and $m$. Note that this phenomenon results into a tradeoff between the generalization and the runtime efficiency for QCBM with classical kernels. That is, increasing $n$ and $m$ enables the reduced generalization error but increases the runtime cost to compute $\MMD_U^2(\PP_{\btheta}^n||\QQ^m)$, which echoes with the claim in \cite{hinsche2021learnability}. Conversely,  QCBM with quantum kernels can simultaneously obtain both the good generalization and runtime efficacy, supported by Lemma~\ref{lem:equi-QCBM-QGAN}. We remark that this statement does not contradict with  the result in \cite{hinsche2021learnability}, since our setting does not require the statistical query access. Instead, our results show that a careful design of learning strategies and model constructions can enhance the power of QCBMs to advance GLMs, especially when $\QQ$ refers to the tasks in quantum many body physics and quantum information processing  \cite{huang2021quantum}.

\textbf{Remark.} In SM~\ref{appd:coro:DNN-QCBM}, we partially address another long standing problem of quantum  probabilistic models, i.e.,  whether QCBMs are superior to classical GLMs. Briefly, for certain $\QQ$, QCBMs can attain a lower $\inf_{\btheta\in\bm{\Theta}}\MMD^2 (\PP_{\bm{\theta}}||\QQ)$ over restricted Boltzmann machine \cite{hinton2012practical}, which may lead to a better learnability.

\section{Generalization of QGAN}\label{sec:gene-QGAN}
We next analyze the generalization of QGANs for continuous $\QQ$. As shown  in Fig.~\ref{fig:QCBM_scheme}(b), QGAN adopts an $N$-qudit Ansatz $\hat{U}(\btheta)$ in Eq.~(\ref{eqn:ansatz}) to build a generator $G_{\bm{\theta}}(\cdot)$, which maps $\bm{z}$ sampled from a prior distribution $\PP_{\mathcal{Z}}$ to a generated example $\rx\sim \PP_{\btheta}$, i.e., $\rx := G_{\bm{\theta}}(\bm{z})\in \mathbb{R}^{d^N}$. The $j$-th entry of $\rx$  is
\begin{equation}\label{eqn:QGAN}
	  \rx_{j}=\Tr(\Pi_j \hat{U}(\bm{\theta})\rho_{\bm{z}}\hat{U}(\bm{\theta})^\dagger), ~\forall j \in [d^N],
\end{equation} 
where $\rho_{\bm{z}}$ refers to the encoded quantum state of  $\bm{z}$. Given $n$ reference examples $\{\rx^{(i)}\}_{i=1}^n$ produced by $G_{\bm{\theta}}(\cdot)$, its  empirical distribution is $\PP_{\btheta}^n(d\rx)=  \sum_{i=1}^n\delta_{\rx^{(i)}}(d\rx)/n$. Similarly, we define $\QQ^m$ as the empirical distribution of $\QQ$ with $m$ true examples. In the training process, QGAN updates $\btheta$ to minimize the empirical MMD loss in Eq.~(\ref{eqn:opt_MMD_empirical}), i.e.,  the optimal solution satisfies $\btheta^*=\arg \min \MMD_U^2(\PP_{\btheta}^n||\QQ^m)$.  Note that QGANs can be applied to estimate both discrete and continuous distributions. When $\QQ$ is discrete, the mechanism of QGANs is \textit{equivalent to} QCBMs (refer to SM \ref{apped:sum_GLMs} \& \ref{appendix:MMD_opt} for more elaborations about QGANs). Due to this reason, here we only focus on applying QGANs to estimate the continuous distribution $\QQ$. 

When QGANs in Eq.~(\ref{eqn:QGAN}) are designed to estimate the continuous distribution $\QQ$, their generalization can only be measured by $\mathfrak{R}^C$ in Definition \ref{def:generalization}. In this scenario, it is of paramount importance of unveiling how the generalization of QGANs depends on uploading methods and the structure information Ansatz. The following theorem makes a concrete step toward this goal.  
\begin{thm}\label{thm:Gene-QGAN}
Assume the kernel $k(\cdot,\cdot)$ is $C_3$-Lipschitz. Suppose that the employed quantum circuit $\hat{U}(\bm{z})$ to prepare $\rho_{\bm{z}}$  containing in total $N_{ge}$ parameterized gates and each gate  acting on at most $k$ qudits. Following notations in Eqs.~(\ref{eqn:opt_MMD_empirical}) and (\ref{eqn:gene_error_bound-emp}), with probability at least $1-\delta$, the generalization error of QGANs,    $\mathfrak{R}^C $ in Eq.~(\ref{eqn:QGAN}) is upper bounded by 
	\begin{equation}\label{eqn:thm2-1}
8\sqrt{ \frac{8C_2^2(n+m)}{nm}\ln \frac{1}{ \delta}  }  +   \frac{48}{n-1} + \frac{144d^{k} \sqrt{N_{gt}+N_{ge}}}{n-1} C_4,
	\end{equation}
where $C_2=\sup_{\rx} k(\rx, \rx)$,  $C_4= N \ln (441d C_3^2nN_{ge}N_{gt}) + 1 $.  
\end{thm}
\begin{proof}[Proof sketch]
Here we only deliver the central idea and postpone the whole proof to SM \ref{appendix:gene_QGAN}.
	
For clearness, we denote $\hat{\btheta}^{(n,m)}$ as $\hat{\btheta}$. Define $\mathcal{E}(\btheta)=\MMD_U^2(\PP_{\btheta}^n||\QQ^m)$ and $\mathcal{T}(\btheta)=\MMD^2(\PP_{\btheta}||\QQ)$. Then the generalization error  is upper bounded by
\begin{eqnarray}
	\mathfrak{R}^C  = && \mathcal{E}(\hat{\btheta})-\mathcal{E}( \hat{\btheta}) +\mathcal{T}(\hat{\btheta})-\mathcal{T}( \btheta^*) \nonumber\\
	 \leq && |\mathcal{T}( \hat{\btheta}) -\mathcal{E}( \hat{\btheta})| + |\mathcal{T}(\btheta^*) -\mathcal{E}(\btheta^*)|,
\end{eqnarray} 
where the first inequality employs the definition of $\hat{\btheta}$ with $\mathcal{E}(\btheta^*)\geq \mathcal{E}(\hat{\btheta})$ and the second inequality uses the property of absolute value function. In this regard, we derive the probability for $\sup_{\btheta \in \Theta} |\mathcal{T}(\btheta) -\mathcal{E}(\btheta)| < \epsilon$,  which in turn can achieve the upper bound of $\mathfrak{R}^C$. 

According to the explicit form of the MMD loss and Jensen inequality, $\sup_{\btheta\in \bm{\Theta}}|\mathcal{E}(\btheta)  - \mathcal{T}(\btheta)|$ satisfies
\begin{align} 
	& \sup_{\btheta\in \bm{\Theta}}\left|\MMD_U^2(\PP_{\btheta}^n||\QQ^m) - \MMD^2(\PP_{\btheta}||\QQ)\right| \nonumber\\
	\leq & \underbrace{\sup_{\btheta\in \bm{\Theta}}\Big|\EE_{\ry,\ry'}(k(\ry,\ry')) - \frac{\sum_{j\neq j'}k(\ry^{(j)}, \ry^{(j')})}{m(m-1)}\Big|}_{T1} \nonumber\\
	 +  &\underbrace{\sup_{\btheta\in \bm{\Theta}}\Big| \EE_{\bm{z},\bm{z}'}(k(G_{\btheta}(\bm{z}),G_{\btheta}(\bm{z}')))- \frac{\sum_{i\neq i'}k(G_{\btheta}(\bm{z}^{(i)}), G_{\btheta}(\bm{z}^{(i')}))}{n(n-1)}\Big|}_{T2} \nonumber\\
	 + & 2\underbrace{\sup_{\btheta\in \bm{\Theta}}\Big|\EE_{z, y}(k(G_{\btheta}(\bm{z}),\ry))- \frac{\sum_{i\in[n],j\in[m]}k(G_{\btheta}(\bm{z}^{(i)}), \ry^{(j)})}{mn}\Big|}_{T3}. \nonumber
\end{align} 
In other words, to upper bound $\mathfrak{R}^C$, we should separately derive the upper bounds of the terms $T1$, $T2$, and $T3$.  Specifically, supported by concentration  inequality \cite{mendelson2003few}, we obtain with probability at least $1- 3 \delta_{T3}$, 
\begin{equation}
\mathfrak{R}^C \leq 2 (\epsilon_1 + 2\epsilon_2 + 4\epsilon),
\end{equation}
where $\epsilon = \sqrt{\frac{8C^2(n+m)}{nm\log 1/ \delta_{T3}}}$, $\epsilon_1=\EE(T2)$, and $\epsilon_2=\EE(T3)$.

The above inequality indicates that the generalization error of QGAN depends on $\epsilon_1$ and $\epsilon_2$.  Namely, by leveraging the covering number---a tool developed in statistical learning theory, we prove $\epsilon_1 \leq \frac{8}{n-1} + \frac{24\sqrt{d^{2k}(N_{ge}+N_{gt})}}{n-1}\left(1+  N\ln (441d C_3^2(n-1)N_{ge}N_{gt}) \right)$  and $\epsilon_2\leq   \frac{8}{n} + \frac{24\sqrt{d^{2k}(N_{ge}+N_{gt})}}{n}\left(1+  N\ln (441d C_3^2nN_{ge}N_{gt}) \right)$. Taken together, with probability $1-3\delta_{T3}$, the generalization error of QGANs is upper bounded by $8\sqrt{ \frac{8C_2^2(n+m)}{nm}\ln \frac{1}{ \delta}  }  +   \frac{48}{n-1} + \frac{144d^{k} \sqrt{N_{gt}+N_{ge}}}{n-1} C_4$. 

\end{proof}
\noindent The results of Theorem \ref{thm:Gene-QGAN} convey four-fold implications. First, according to the first term in the right hand-side of Eq.~(\ref{eqn:thm2-1}),  to achieve  a tight upper bound of  $\mathfrak{R}^C$,  the ratio between the number of examples sampled from $\QQ$ and $\PP$ should satisfy $m/n=1$. In this case, $\mathfrak{R}^C$  linearly decreases and finally converges to zero with the increased $n$ or $m$. Second, $\mathfrak{R}^C$ \textit{linearly} depends on the kernel term $C_2$,  \textit{exponentially} depends on $k$ in Eq.~(\ref{eqn:ansatz}), and  \textit{sublinearly} depends on the number of trainable quantum gates $N_{gt}$. These observations underpin the importance of controlling the expressivity of the adopted Ansatz and selecting proper kernels to ensure both the good learning performance and generalization of QGANs. Third, the way of preparing $\rho_{\bm{z}}$ is a determinant factor in generalization of QGANs, which implies the prior distribution $\PP_{\mathcal{Z}}$ and the number of encoding gates $N_{ge}$ should be carefully designed.  Last, the explicit dependence on the architecture of Ansatz connects the generalization of QGANs with their trainability, i.e.,    a large $N_{gt}$ or $k$ may induce barren plateaus in training QGLMs \cite{mcclean2018barren,kieferova2021quantum} and results in an inferior learning performance. Meanwhile, it also leads to a degraded generalization error bound.

\textbf{Remark.} Most kernels such as RBF, linear, and Mat\'{e}rn kernels satisfy the Lipchitz condition  \cite{Lederer2019}. Meanwhile, Theorem \ref{thm:Gene-QGAN} can be efficiently extended to noisy settings by using the contraction properties of noisy channels explained in Ref.~\cite{du2021efficient}.

The derived bound in Theorem \ref{thm:Gene-QGAN} is succinct and can be directly employed to quantify the generalization of QGANs with the specified Ansatz. The following corollary quantifies $\mathfrak{R}^C$ of QGANs with two typical Ans\"atze in  VQAs, i.e., hardware-efficient Ansatz and quantum approximate optimization Ansatz.
 
\begin{coro}\label{coro:gene_ansatz}
Following notations in Theorem \ref{thm:Gene-QGAN}, when QGAN is realized by the hardware-efficient Ansatz and quantum approximate optimization Ansatz with $L$ layers, $\mathfrak{R}^C $ is upper bounded by \[\tilde{\mathcal{O}}(C_2\sqrt{\frac{n+m}{(nm)}}+ \frac{1}{n}+\frac{\sqrt{N(L_E+3L)}(N+1)}{n-1}),\] and 
\[\tilde{\mathcal{O}}(C_2\sqrt{\frac{(n+m)}{(nm)}}   +   \frac{1}{n-1} + \frac{2^N \sqrt{(N+1)(L+L_E)}(N + 1)}{n-1}).\]  
\end{coro} 
\begin{proof}[Proof of Corollary \ref{coro:gene_ansatz}]
	
Here we separately analyze the generalization ability of QGANs with two typical classes of Ans\"atze, i.e., the hardware-efficient Ans\"atze and the quantum approximation optimization Ans\"atze.

\textit{Hardware-efficient Ans\"atze.}  An $N$-qubits hardware-efficient Ansatz is composed of $L$ layers, i.e., $U(\bm{\theta})=\prod_{l=1}^LU(\bm{\theta}^l)$ with $L\sim poly(N)$, where $U(\bm{\theta}^l)$ is composed of parameterized single-qubit gates and fixed two-qubit gates. In general, the topology of $U(\bm{\theta}^l)$ for any $l\in[L]$ is the same and each qubit interacts with at least one parameterized single-qubit gate and two qubits gates.  Mathematically, we have $U(\bm{\theta}^l)=(\otimes_{i=1}^N U_s)U_{eng}$ with $U_s=R_Z(\beta)R_Y(\gamma)R_Z(\nu)$ being realized by three rotational qubit gates and $ \gamma, \beta, \nu\in[0, 2\pi)$. The number of two-qubit gates in each layer is set as $N$ and the connectively of two-qubit gates aims to adapt to the topology restriction of quantum hardware. The entangled layer $U_{eng}$ contains two-qubit gates, i.e., CNOT gates, whose connectivity adapts the topology of the quantum hardware.  

An example of $4$-qubit QGAN with  hardware-efficient Ansatz is illustrated in the upper panel of Fig.~\ref{fig:Ansatz-QNN}. Under this setting, when both the encoding unitary and the trainable unitary adopt the hardware-efficient Ansatz, we have $k=2$, $d=2$, $N_{ge}=L_{E}N$, $N_{gt}=L(3N)$, and $N_{g}=L(3N+N) = 4L$.  Based on the above settings, we achieve the generalization error of ab $N$-qubit QGAN  with the hardware-efficient Ans\"atze, supported by Theorem \ref{thm:Gene-QGAN}, i.e., with probability at least $1-\delta$,      
	\begin{eqnarray}
		\mathfrak{R}^C  \leq && 
8\sqrt{ \frac{8C_2^2(n+m)}{nm}\ln \frac{1}{ \delta}  }  +   \frac{48}{n-1} \nonumber\\
&& + \frac{576 \sqrt{N(L_E+3L)}}{n-1}(N\ln(1323d C_3^2nN^2L_EL)+1), \nonumber
	\end{eqnarray}  
where $C_2=\sup_{\rx} k(\rx, \rx)$.

  \begin{figure}
 	\centering
 	\includegraphics[width=0.43\textwidth]{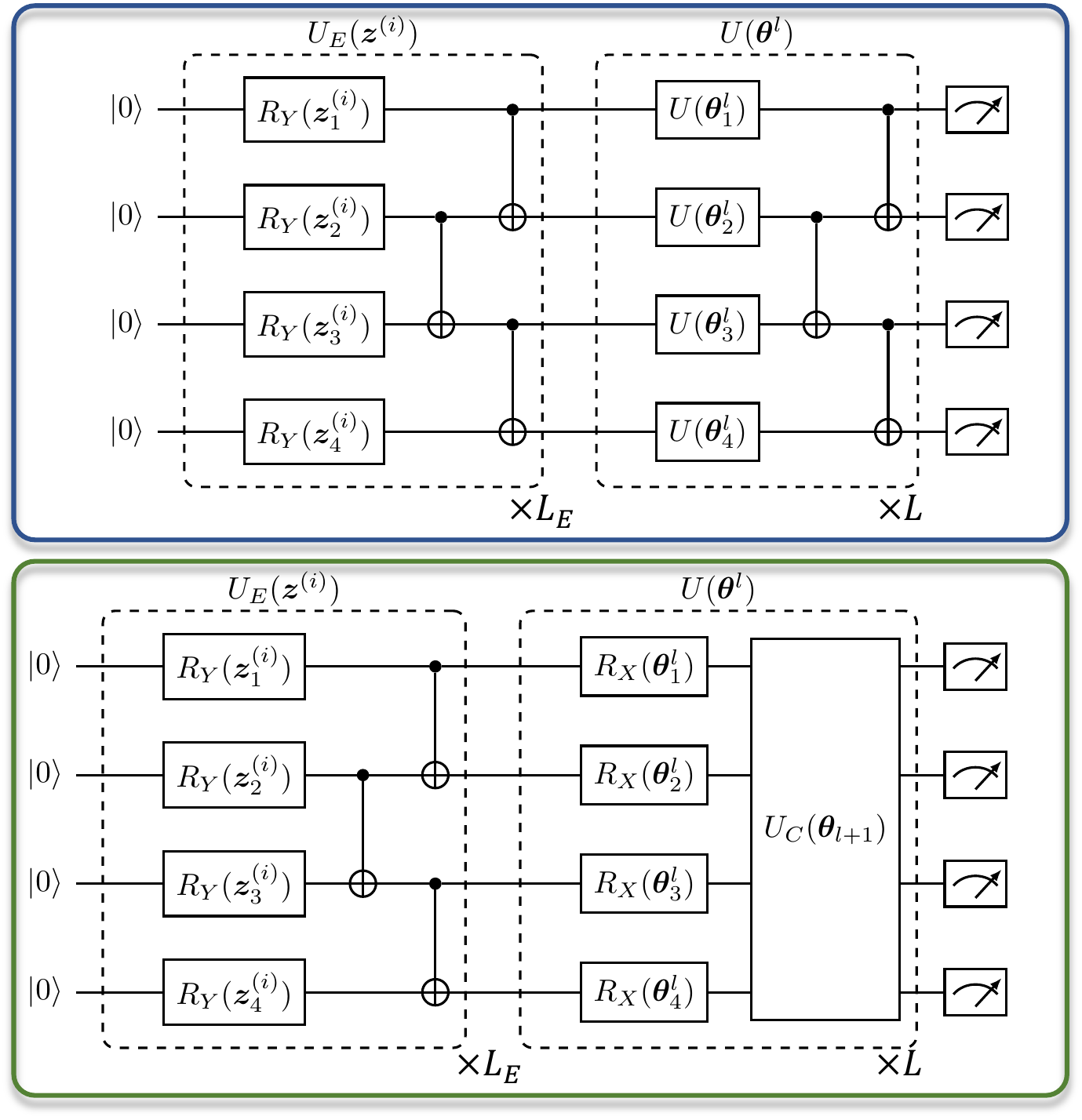}
 	\caption{\small{\textbf{Illustration of QGANs with different Ans\"atze.} The upper panel presents that both the encoding method and the trainable unitary of QGANs employ the  hardware-efficient  Ans\"atze. The lower panel presents a class of QGANs such that the encoding method uses the  hardware-efficient  Ans\"atze and the trainable unitary is implemented by the quantum approximate optimization Ans\"atze. }}
 	\label{fig:Ansatz-QNN}
 \end{figure} 
  
\textit{Quantum approximate optimization Ans\"atze.}  The lower panel of Fig.~\ref{fig:Ansatz-QNN} plots the quantum approximate optimization Ans\"atze. The mathematical expression of this Ansatz takes the form $U(\bm{\theta})=\prod_{l=1}^LU(\bm{\theta}^l)$, where the $l$-th layer $U(\bm{\theta}^l)=U_B(\bm{\theta}^l)U_C(\bm{\theta}^l)$ is implemented by the driven Hamiltonian $U_B(\bm{\theta}^l)=\otimes_{i=1}^N R_X(\bm{\theta}^l_{i})$ and the target Hamiltonian $U_C(\bm{\theta}^l)=\exp(-i\btheta^l_{l+1}H_C)$ with $H_C$ being a specified Hamiltonian. Under this setting, when the encoding unitary is constructed by the hardware-efficient Ansatz and the trainable unitary is realized by the quantum approximate optimization Ansatze, we have $k=N$, $d=2$, $N_{ge}=L_{E}N$, and $N_{gt}=L(N+1)$.  Based on the above settings, we achieve the generalization error of an $N$-qubit QGAN with the quantum approximate optimization Ans\"atze, supported by Theorem \ref{thm:Gene-QGAN}, i.e., with probability at least $1-\delta$, 
  	\begin{eqnarray}
\mathfrak{R}^C\leq  && 8\sqrt{ \frac{8C_2^2(n+m)}{nm}\ln \frac{1}{ \delta}  }  +   \frac{48}{n-1} \nonumber\\ 
&&+ \frac{144\times 2^N \sqrt{(N+1)(L+L_E)}}{n-1} \nonumber\\
&& \times (N \ln (441d C_3^2nLL_EN(N+1)) + 1), \nonumber
  	\end{eqnarray} 
where $C_2=\sup_{\rx} k(\rx, \rx)$.  
\end{proof}

\section{Applications of QGLMs with merits}\label{sec:application-QGLM}
A crucial message conveyed by Theorems \ref{thm:Gene-QCBM} and \ref{thm:Gene-QGAN} is that when $\QQ \in \PP_{\bm{\Theta}}$, the   distance between the learned $\PP_{\btheta}$ and $\QQ$ is continuously decreased by increasing $n$ and $m$. This  evidence contributes to theoretically study the merits of QGLMs in many  practical problems such as quantum state preparation and parameterized Hamiltonian learning. 

\subsection{Quantum state preparation}

Quantum state preparation is a crucial subroutine in quantum computing and quantum sensing, since its efficiency determines the achievable runtime speedups of quantum machine learning and Hamiltonian simulation algorithms, and improves the precision of the device. Nevertheless, theoretical results indicate that exactly loading the generic discrete distribution with $d^N$ dimensions to an $N$-qudit state requests $O(d^N)$ gates \cite{plesch2011quantum,zhang2022quantum}, which prevents any merits. Moreover, there is no deterministic method of encoding continuous distribution into the multi-qudit system. QGLMs open the door to efficiently encode both discrete and continuous distributions into quantum circuits with certain estimation error \cite{zoufal2019quantum}.  Theorems \ref{thm:Gene-QCBM} and \ref{thm:Gene-QGAN} establish the theoretical foundation about such an error. Specifically, for $\QQ \in \PP_{\Theta}$, large examples promise a low estimation error and the  learned QGLM can be easily realized and reused. For $\QQ \notin \PP_{\Theta}$, the estimation error could be considerable, even though $n$ and $m$ are infinite. In this view, the power of QGLMs in quantum state preparation highly depends on their expressivity. 

\subsection{Parameterized Hamiltonian Learning}
QGLMs may advance GLMs in  parameterized Hamiltonian learning (PHL). Define an $N$-qubit parameterized Hamiltonian as $H(\bm{a})$ and the corresponding ground state as $\ket{\phi(\bm{a})}$, where $\bm{a}$ is the interaction parameter sampled from a prior   distribution $\mathbb{D}$. PHL aims to use $m$ training samples $\{\bm{a}^{(i)}, \ket{\phi(\bm{a}^{(i)})}\}_{i=1}^m$ to approximate the distribution of the ground states for $H(\bm{a})$ with $\bm{a}\sim \mathbb{D}$, i.e., $\ket{\phi(\bm{a})}\sim \QQ$. If the estimation error is low, then the trained model can prepare the ground state of $H(\bm{a}')$ for an unseen parameter $\bm{a}'\sim \mathbb{D}$. This property can be used to explore many crucial behaviors in condensed-matter systems.   Concisely, there exists an $N$-qubit Hamiltonian $H(\bm{a})$ and a distribution $\mathbb{D}$ such that the formed ground state distribution $\QQ$ can be efficiently estimated by QGLMs with $O(poly(N))$ trainable parameters but is computational hardness for GLMs with $O(poly(N))$ trainable parameters.

 \begin{figure*}[!htp]
	\centering
\includegraphics[width=0.98\textwidth]{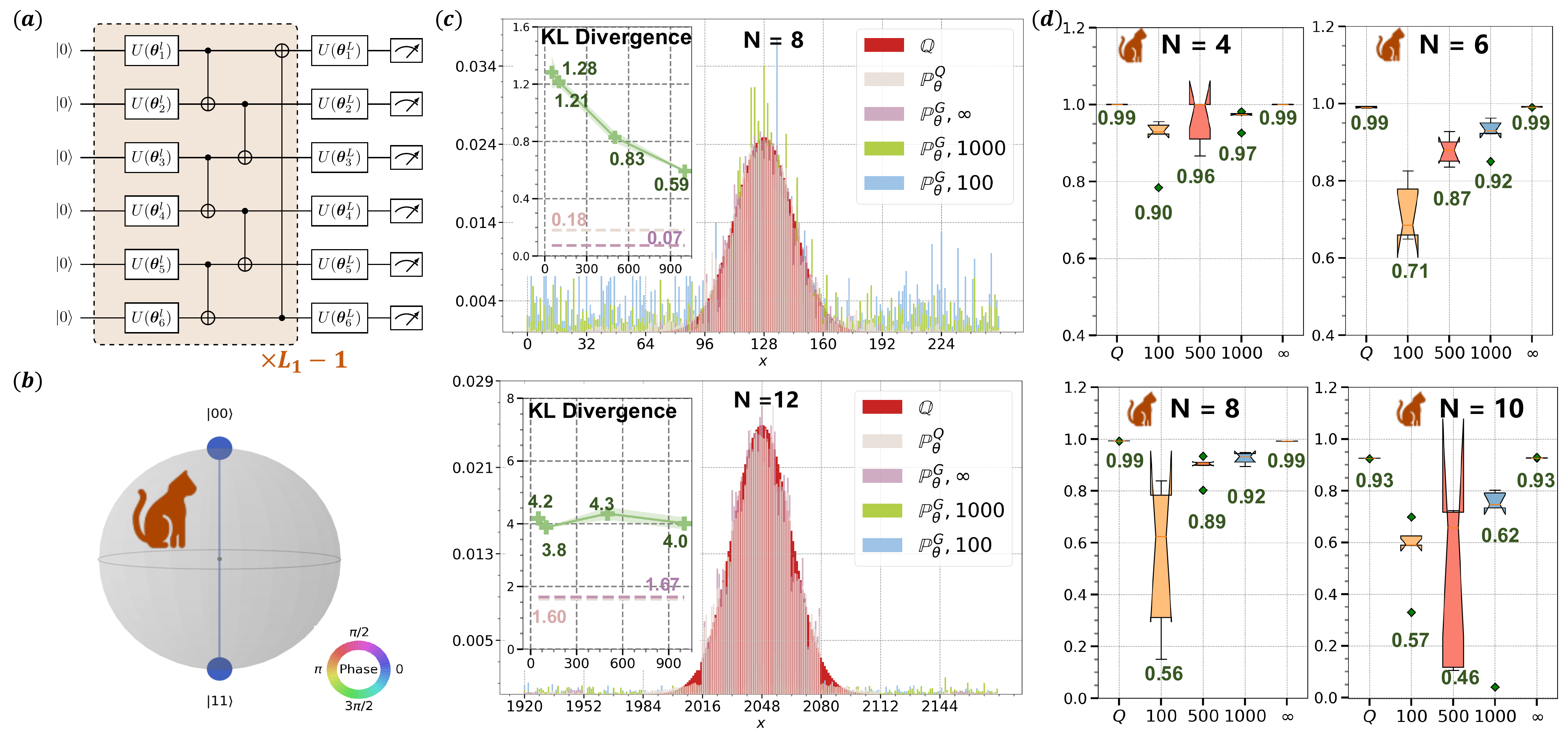}
\caption{\small{\textbf{Simulation results of QCBM.} (a) The implementation of QCBMs when $N=6$. The label `$L_1-1$' refers to repeating the architecture highlighted by the brown color $L_1-1$ times. The gate $U(\btheta_i^l)$ refers to the $\RY$ gate applied on the $i$-th qubit in the $l$-th  layer of Ansatz $\hat{U}(\btheta)$.  (b) The visualization of a two-qubit GHZ state. (c) The upper and lower panels separately show the simulation results of QCBMs in the task of estimating the discrete Gaussian distribution when $N=8$ and $N=12$. The labels `$\QQ$', $\PP_{\btheta}^Q$, `$\PP_{\btheta}^G,n$', stand for the target distribution, the output of QCBM with the quantum kernel, and the output of QCBM with the RBF kernel and  $n$ samples, respectively.   The inner plots evaluate the statistical performance of QCBM through KL divergence, where the x-axis labels the number of examples $n$. (d) The four box-plots separately show the simulation results of QCBM in the task of approximating $N$-qubit GHZ state with $N=4, 6, 8, 10$. The y-axis refers to the fidelity. The x-axis refers to the applied kernels in QCBM, where the label `$Q$' represents the quantum kernel and the rest four labels refers to the RBF kernel with  $n$ samples.     }}
\label{fig:dist_main_QCBM}
\end{figure*} 

 The separated power between QGLMs and GLMs relies on the following lemma.
\begin{lem}\label{lem:append:11}
	Suppose that $\QQ$ refers to the distribution of the   ground states for  parameterized Hamiltonians   $H(\bm{a})$ with $\bm{a}\sim \mathbb{D}$. Under the quantum threshold assumption, there exists a  distribution $\QQ$ that can be efficiently represented by QGLMs but is computationally hard for GLMs.     
\end{lem} 
\begin{proof}[Proof sketch]
The proof is provided in SM \ref{append:strm}. Conceptually, the proof is established on the results of quantum random circuits, which is widely believed to be classically computationally hard and in turn can be used to demonstrate quantum advantages on NISQ devices \cite{boixo2018characterizing,zhu2021quantum}. Namely, Ref.~\cite{aaronson2017complexity} proposed a Heavy Output Generation (HOG) problem to separate the power between classical and quantum computers. That is,  under the quantum threshold assumption, there do not exist classical samples that can spoof  the $k$ samples output by a random quantum circuit with success probability at least $0.99$ \cite{aaronson2017complexity}, while   quantum computers can solve the HOG problem with high success probability.   

Based on this observation, we connect $\ket{\phi(\bm a)}$  with the output state of random quantum circuits. To this end, we prove that there exists a ground state $\ket{\phi(\bm a)}$ of a Hamiltonian $H(\bm a)$, which can be efficiently prepared by quantum computers but is computationally hard for classical algorithms.  
\end{proof}
\noindent It indicates that in PHL, the expressivity of QGLMs outperforms GLMs with  $\QQ\in \PP_{\bm{\Theta}}^Q$ and $\QQ \notin \PP_{\bm{\Theta}}^C$. When the number of trainable parameters of QGLMs scales with  $O(poly(N))$,  there may exist a kernel leading to $\inf_{\bm\theta\in \bm\Theta} \MMD^2 (\PP_{\bm\theta}^Q ||\QQ)\rightarrow 0$. Considering that $\PP_{\bm\theta}^Q$ and $\QQ$ are intractable, the training of  QGANs amounts to minimizing $\MMD_U^2 (\PP_{\bm\theta}^Q ||\QQ)$.  Theorems \ref{thm:Gene-QCBM} and \ref{thm:Gene-QGAN} imply that when the empirical loss tends to zero and $n$ and $m$ are sufficient large, the generalization error vanishes and thus the estimated distribution recovers the target distribution with $ \PP_{\bm\theta^{(n,m)}}^Q \approx \PP_{\bm\theta}^Q\approx \QQ$. Conversely, the result $\QQ \notin \PP_{\bm{\Theta}}^C$ for GLMs means that even though the empirical loss and the generalization error are zero, the estimated distribution $ \PP_{\bm\theta^{(n,m)}}^C$ fails to recover $\QQ$.

\section{Numerical simulations}\label{sec:numerical}
In this section, we apply QCBMs and QGANs to accomplish the distribution preparation and quantum state approximation tasks. More implementation details and simulation results are deferred to  SM \ref{appendix:sim}.

\subsection{Gaussian distribution preparation by QCBM}
The first task is applying QCBMs to prepare the discrete Gaussian distribution $\mathsf{N}(N, \mu, \sigma)$, where $N$ specifies the range of events with $\rx\in [2^N]$, and $\mu$ and $\sigma$ refer to the mean and variance, respectively. An intuition of $\mathsf{N}(N, \mu, \sigma)$  is shown in Fig.~\ref{fig:dist_main_QCBM}(c), labeled by $\QQ$. The hyper-parameter settings are as follows. For all simulations, we fix $\mu=1$ and $\sigma=8$. The qubit count is set as $N=12$.  The hardware-efficient Ansatze is employed to construct QCBM with $L_1=8$ for $N=8$ ($L_1=12$ for $N=12$). An intuition is  in illustrated in Fig.~\ref{fig:dist_main_QCBM}(a). The quantum kernel and RBF kernel are adopted to compute the MMD loss. For RBF kernel, the number of samples is set as $n=100$, $1000$, and $\infty$. The maximum number of iterations is $T=50$.  For each setting, we repeat the training $5$ times to collect the statistical results.

\begin{figure*}[!htp]
	\centering
\includegraphics[width=0.98\textwidth]{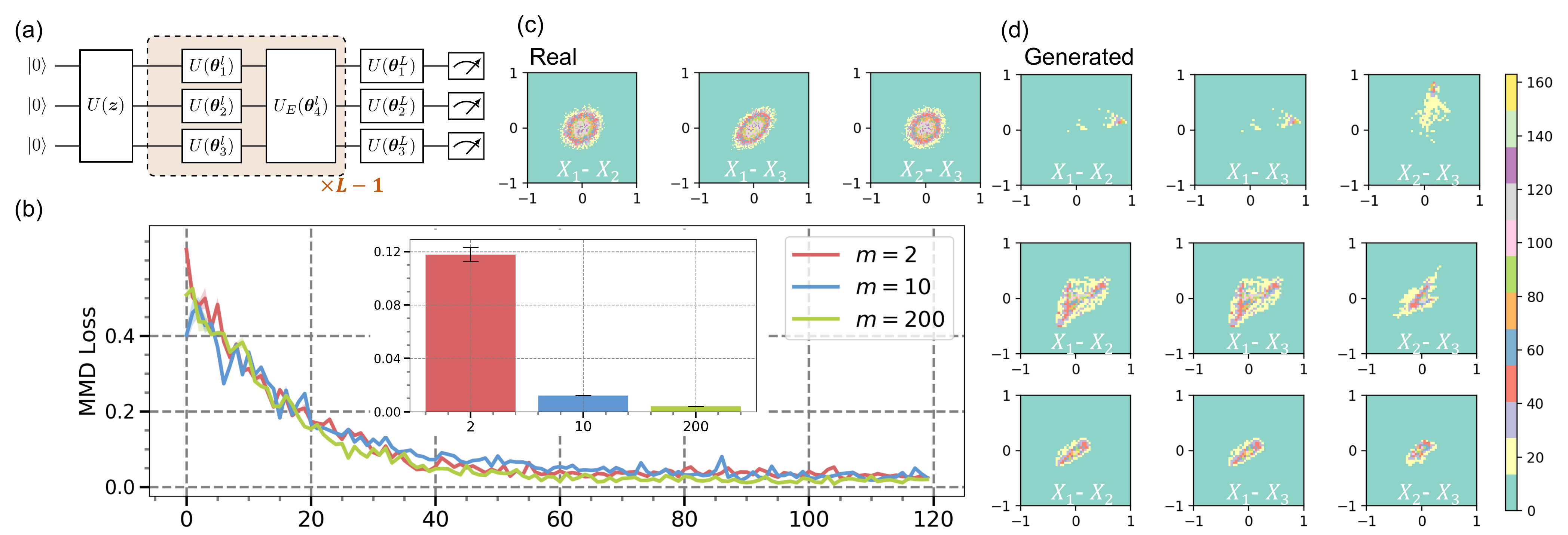}
\caption{\small{\textbf{Simulation results of QGANs.}   (a) The implementation of QGANs when the number of qubits is $N=3$. $U(\bm{z})$ refers to the encoding circuit to load the example $\bm{z}$. The meaning of `$L_1-1$' is identical to the one explained in Fig.~\ref{fig:QCBM_scheme}. The gate $U(\btheta_i^l)$ refers to the $\RY\RZ$ gates applied on the $i$-th qubit in the $l$-th  layer of $\hat{U}(\btheta)$. (b) The outer plot shows the training loss of QGANs with varied settings of $m$. The x-axis refers to the number of iterations. The inner plot shows the generalization property of trained QGANs by evaluating MMD loss. (c) The visualization of the exploited 3D Gaussian distribution. The label `$X_a$-$X_b$' means projects the 3D Gaussian into the $X_a$-$X_b$ plane with $a,b$ belonging to x, y, z-axis. (d) The generated data sampled from the trained QGAN with varied settings of $m$. From upper to lower panel, $m$ equals to $2$, $10$, and $200$, respectively. }}
\label{fig:3D-gau}
\end{figure*}

The simulation results of QCBMs are illustrated in Fig.~\ref{fig:dist_main_QCBM}(c). The two outer plots exhibit the approximated distributions under different settings. In particular, for both $N=8, 12$, the approximated distribution generated by QCBM with the quantum kernel well approximates $\QQ$. In the measure of KL divergence, the similarity of these two distributions is $0.18$ and $1.6$ for $N=8, 12$, respectively. In contrast, when the adopted kernels are classical and the number of measurements is finite, QCBMs encounter the inferior performance. Namely, by increasing $n$ and $m$ from $50$ to $1000$, the KL divergence between the approximated distribution and the target distribution only decreases from $1.28$ to $0.59$ in the case of $N=8$. Moreover, under the same setting, the KL divergence does not manifestly decrease when $N=12$, which requires a larger $n$ and $m$ to attain a good approximation  as suggested by Theorem \ref{thm:Gene-QCBM}. This argument is warranted by the numerical results with the setting $n=m\rightarrow \infty$, where the achieved KL divergence is comparable with QCBM with the quantum kernel. Nevertheless,  the runtime complexity of QCBMs with the classical kernel polynomially scales with $n$ and $m$. According to Lemma \ref{lem:equi-QCBM-QGAN}, under this scenario, QCBMs with the quantum kernel embraces the runtime advantages.

 \subsection{GHZ state approximation by QCBM}
We apply QCBMs to accomplish the task of preparing GHZ states, a.k.a., ``cat states'' \cite{nielsen2010quantum}.   An intuition is depicted in Fig.~\ref{fig:dist_main_QCBM}(b).  The choice of GHZ states is motivated by their importance in quantum information.   The formal expression of an $N$-qubit GHZ state is $\ket{\GHZ}=(\ket{0}^{\otimes N} + \ket{1}^{\otimes N})/\sqrt{2}$. The hyper-parameter settings are as follows. The number of qubits is set as $N= 6, 8, 10$ and the corresponding depth is $L_1= 4, 6, 8, 10$. The quantum kernel and RBF kernel are adopted to compute the MMD loss. For RBF kernel, the number of samples is set as $n=100$, $1000$, and $\infty$.  The maximum number of iterations is $T=50$.  For each setting, we repeat the training $5$ times to get a better understanding of the robustness of the results.   The other settings are identical to those used in the task of preparing Gaussian distributions. 
 
The simulation results, as illustrated in Fig.~\ref{fig:dist_main_QCBM}(d),  indicate that QCBMs with quantum kernels outperforms RBF kernels when $n$ and $m$ are finite. This observation becomes apparent with an increased $N$. For all settings of $N$, the averaged fidelity between the generated states of QCBMs with the quantum kernel and the target $\ket{\GHZ}$ is above $0.99$, whereas the obtained averaged fidelity for QCBMs with the RBF kernel is $0.46$ for $N=10$ and $n=m=100$.  Meanwhile, as with the prior task, RBF kernel attain a competitive performance with the quantum kernel unless $n=m\rightarrow \infty$, while the price to pay is an unaffordable computational overhead.

\subsection{3D-correlated Gaussian preparation by QGAN} 
The last tasks is using QGANs to prepare  3D correlated Gaussian distributions with varied settings. The target distribution $\QQ$ is a 3D correlated Gaussian distribution centered at $\mu = (0, 0, 0)$ with covariance matrix $\sigma=\big(\begin{smallmatrix}
  0.5 & 1 & 0.25\\
  0.1 & 0.5 & 0.1\\
  0.25 & 0.1 & 0.5
\end{smallmatrix}\big)$. The sampled examples from $\QQ$ are visualized in Fig.~\ref{fig:3D-gau}(a).  Here we use a variant of the style-QGAN proposed by \cite{bravo2021style} to accomplish the task. Two key modifications in our protocol are  constructing the quantum generator and replacing the trainable discriminator by MMD loss. Different from the original proposal applying the re-uploading method, the modified quantum generator first uploads the prior example $\bm{z}$ using $U(\bm{z})$ followed by the Ansatz $U(\btheta)$. Such a modification facilitates the analysis of the generalization behavior of QGANs as claimed in Theorem \ref{thm:Gene-QGAN}.  The hyper-parameter settings are as follows.  The number of reference samples $n$ ranges from $2$ to $200$ and we keep $n=m$. The layer depth of $G_{\btheta}(\cdot)$ is set as $L\in\{2,4, 6, 8\}$. Each setting is repeated with $5$ times to collect the statistical results. 

 The simulation results are exhibited in Figs.~\ref{fig:3D-gau}(b)-(c). For illustration, we depicts the generated distribution of the trained QGANs in Fig.~\ref{fig:3D-gau}(b). With increasing $m$, the learned distribution is close to the real distribution. The outer plot in Fig.~\ref{fig:3D-gau}(c) shows that for all settings of $m$, the   empirical $\MMD$ loss, i.e., $\MMD_U(\PP_{\hat{\btheta}^{(n,m)}}^n||\QQ^m)$, converges after $60$ iterations, where the averaged loss is $0.0114$, $0.0077$, and $0.0054$ for $m=2, 10, 200$, respectively. The inner plot measures the expected $\MMD$ loss, i.e., the trained QGANs are employed to generate new $10000$ examples and then evaluate $\MMD_U$ to estimate $\MMD$. The averaged expected $\MMD$ loss for $m=2, 10, 200$ is $0.1178$, $0.0122$, and $0.0041$ respectively.  These observations echo with Theorem \ref{thm:Gene-QGAN}, where a large $m$ allows a better generalization ability.

 \section{Discussions}\label{sec:discuss}
 We provide a succinct and direct way to compare generalization of QGLMs with different Ans\"atze, which deepens our understanding about the capabilities of QGLMs and benefit the design of advanced QGLMs. For QCBMs, Theorem \ref{thm:Gene-QCBM}  unveils that quantum kernels can greatly benefit their generalization and reduce computational overhead over classical kernels when the target distribution is quantum. Theorem \ref{thm:Gene-QGAN} hints that the generalization error of QGANs has the explicit dependence  on the qudits count, the structure information of the employed Ans\"atze,  the adopted encoding method, and the choice of prior distribution. These   results suggest that a possible way to enhance the learning performance of QGLMs is devising novel Ans\"atze  via quantum circuit architecture design techniques \cite{amaro2022filtering,bilkis2021semi,du2020quantum,linghu2022quantum,kuo2021quantum,zhang2020differentiable}.      Although the attained theoretical results do not exhibit the generic exponential advantages of QGLMs, we clearly show that their potentials in quantum state preparation and parameterized Hamiltonian learning.

The developed techniques in this study are general and provide a novel approach to theoretically investigate  the power of QGLMs. For instance, a promising direction is uncovering the generalization of other QGLMs. Furthermore, it is intriguing to explore the generalization of QGLMs with other loss functions  \cite{briol2019statistical,chakrabarti2019quantum}. Another  important future direction will be identifying how to use QGLMs to gain substantial quantum advantages for practical applications, e.g.,  quantum many body physics, quantum sensing, and quantum information processing. Last, the entangled relation between expressivity and generalization in QGLMs queries a deeper understanding from each side.

\newpage   
\clearpage 
\renewcommand{\appendixname}{SM}
\appendix 
\onecolumngrid

\begin{center}
\large{\textbf{Supplementary Material: ``Power of Quantum Generative Learning''}}	
\end{center}

\section{Schematic of QGANs in the discrete and continuous settings}\label{apped:sum_GLMs}

For the purpose of elucidating, in this section, we first introduce the basic theory of GANs and QGANs, especially for QGANs with MMD loss, and then demonstrate the equivalence of QCBMs and QGANs in the discrete setting when the loss function is specified to be MMD.

\medskip 
\textbf{Basic theory of (classical) GANs and QGANs when $\QQ$ is continuous.}  
  The fundamental mechanism of GAN \cite{goodfellow2014generative} and its variations \cite{arjovsky2017wasserstein, mirza2014conditional,zhang2017stackgan,makhzani2015adversarial} is as follows. GAN sets up a  two-players game: the generator  $G$ creates data that pretends to come from the real data distribution $\QQ$ to fool the discriminator $D$, while $D$ tries to distinguish the fake generated data from the real training data. Mathematically,  $G$ and $D$ corresponds to two a differentiable functions. In particular, the input of $G$ is a latent variable $\bm{z}$ and its output is $\rx$, i.e.,  $G:G(\bm{z}, \bm{\theta})\rightarrow \rx$ with $\bm{\theta}$ being trainable parameters for $G$. The role of the latent variable $\bm{z}$ is  ensuring  GAN to be a structured probabilistic model \cite{goodfellow2016deep}. The input of $D$ can either be the generated data $\rx$ or the real data $\ry\sim \QQ$ and its output corresponds to the binary classification result (real or fake), respectively. The mathematical expression of $D$ yields $D : D(\rx, \ry, \bm{\gamma})\rightarrow (0,1) $ with $\bm{\gamma}$ being trainable parameters for $D$. If the distribution  learned by $G$ equals to the real data distribution, i.e., $\PP_{\btheta} = \QQ$, then $D$ can never discriminate between the generated data and the real data and this unique solution    is called Nash equilibrium \cite{goodfellow2014generative}.

To reach the Nash equilibrium, the training process of GANs corresponds to the minimax optimization. Namely, the discriminator $D$ updates  $\bm{\gamma}$  to maximize the classification accuracy, while the generator $G$ updates  $\btheta$ to minimize the classification accuracy. With this regard, the optimization of GAN  follows
\begin{equation}\label{eqn:loss}
\min_{\bm{\theta}}\max_{\bm{\gamma}} \mathcal{L}(D_{\bm{\gamma}}(G_{\bm{\theta}}(\bm{z})),D_{\bm{\gamma}}(\bm{x})):=  \mathbb{E}_{\bm{x}\sim \QQ} [ D_{\bm{\gamma}}( \bm{x})] +\mathbb{E}_{\bm{z}\sim \Pro_{\mathcal{Z}}} [(1-D_{\bm{\gamma}}(G_{\bm{\theta}}(\bm{z}))],
\end{equation}
where  $\QQ$ is the distribution of training dataset, and $\Pro_{\mathcal{Z}}$ is the probability distribution of the latent variable $\bm{z}$. In general, $G_{\bm{\theta}}$ and $D_{\bm{\gamma}}$ are constructed by deep neural networks, and their parameters are updated iteratively using  gradient descent methods \cite{boyd2004convex}.    
 
The key difference between GANs and QGANs is the way of implementing $G_{\btheta}$ and $D_{\bm{\gamma}}$. Particularly, in  QGANs, either $G$, $D$, or both can be realized by variational quantum circuits instead of deep neural networks. The training strategy of QGANs is similar to classical GANs. In this study, we focus on QGANs with MMD loss, which can be treated as the quantum extension of MMD-GAN \cite{dziugaite2015training}. Unlike conventional GANs and QGANs,  MMD-GAN and QGAN  replace a trainable discriminator with   MMD. In this way,  the family of discriminators is substituted with a family $\mathcal{H}$ of test functions, closed under negation, where the optimization of $D$ can be completed with the analytical form. Therefore, the goal of QGANs is finding an estimator minimizing an unbiased  MMD loss, i.e.,
\begin{equation}\label{append:eqn:unbia-MMD}
	\MMD^2_U(\PP_{\btheta}||\QQ):=\frac{1}{n(n-1)}\sum_{i\neq i'}^n k(\rx^{(i)}, \rx^{(i')}) + \frac{1}{m(m-1)}\sum_{j\neq j'}^m k(\ry^{(j)}, \ry^{(j')})- \frac{2}{nm}\sum_{i,j} k(\rx^{(i)}, \ry^{(j)}), 
\end{equation}
where $\rx^{(i)}\sim \PP_{\btheta}$ and $\ry^{(j)}\sim \QQ$.

\medskip
\textbf{Equivalence between QCBMs and QGANs when $\QQ$ is discrete.} In accordance with the explanations in Refs.~\cite{benedetti2019adversarial,zoufal2019quantum}, when QGAN is applied to estimate a discrete distribution $\QQ$ (e.g., quantum state approximation), the quantum generator aims to directly capture the distribution of the data itself. This violates the criteria of implicit generative models, where a stochastic process is employed to draw samples from the underlying data distribution after training. More specifically, when $\QQ$ is discrete, the output of the quantum circuit for both QCBM and QGAN takes the form $\PP_{\btheta}(i)= \Tr(\Pi_i \hat{U}(\bm{\theta})\rho_0 \hat{U}(\bm{\theta})^\dagger)$ in Eq.~(\ref{eqn:QCBM}). The concept `adversarial' originates from the way of optimizing $\btheta$. Instead of using a deterministic distance measure (e.g., KL divergence) as in QCBMs,  QGANs utilize a discriminator $D_{\bm{\gamma}}$, implemented by either trainable parameterized quantum circuit or a neural network, to maximally separate  $\PP_{\btheta}(i)$ from $\QQ$. The behavior of simultaneously update $\btheta$ (to minimize the loss) and $\bm{\gamma}$ (to maximize the loss) is termed as quantum generative adversarial learning. With this regard, when we replace the trainable $D_{\bm{\gamma}}$ by  the deterministic measure MMD, QGAN takes an equivalent mathematical form with QCBM.

\section{Optimization of QGANs with MMD loss}\label{appendix:MMD_opt}
For self-consistency, in this section, we introduce the elementary backgrounds of the optimization of QGLMs with MMD loss. See Ref.~\cite{Gretton2012AKernel} for elaborations. A central concept used in MMD loss is kernels.
\begin{definition}[Definition 2, \cite{schuld2019quantum}]
	Let $\mathcal{X}$ be a nonempty set, called the input set. A function $k:\mathcal{X}\times \mathcal{X}\rightarrow \mathbb{C}$ is called kernel if the Gram matrix $\mathcal{K}$ with entries $\mathcal{K}_{m,m'}=k(\rx^m,\rx^{m'})$ is positive semi-definite. 
\end{definition}

\medskip
\textbf{MMD loss.} Let $k: \mathcal{X} \times  \mathcal{X} \rightarrow \mathbb{R}$ be a Borel measurable kernel on $\mathcal{X}$, and consider the reproducing kernel Hilbert space $\mathcal{H}_k$ associated with k (see Berlinet and Thomas-Agnan [2004]), equipped with inner product $\langle \cdot, \cdot \rangle_{\mathcal{H}_k}$.  Let $\mathcal{P}_k(\mathcal{X})$ be the set of Borel probability measures $\mu$ such that $\int_{\mathcal{X}}\sqrt{k(\rx, \rx)}\mu(d\rx)< \infty$. The kernel mean embedding $\Pi_k(\mu) = \int k(\cdot,y)\mu(dy)$, interpreted as a Bochner integral, defines a continuous embedding from  $\mathcal{P}_k(\mathcal{X})$ into $\mathcal{H}_k$. The mean embedding pulls-back the metric on $\mathcal{H}_k$ generated by the inner product to define a pseudo- metric on $\mathcal{P}_k(\mathcal{X})$ called the maximum mean discrepancy MMD: $\mathcal{P}_k(\mathcal{X})\times \mathcal{P}_k(\mathcal{X})\rightarrow \mathbb{R}_+$, i.e.,
 \begin{equation}
 	\MMD(\PP_1||\PP_2)=\|\Pi_k(\PP_1)-\Pi_k(\PP_2)\|_{\mathcal{H}_k}.
 \end{equation}
 
 The MMD loss has a particularly simple expression that can be derived through an application of the reproducing property ($f(x)=\langle f, k(\cdot, \rx)\rangle_{\mathcal{H}_k}$), i.e.,
\begin{eqnarray}\label{append:eqn:mmd_loss}
	&& \MMD^2(\PP_1||\PP_2):=\left\|\int_{\mathcal{X}}k(\cdot, \rx)\PP_1(d\rx)- \int_{\mathcal{X}}k(\cdot, \rx)\PP_2(d\rx)\right\|_{\mathcal{H}_k}^2 \nonumber\\
	=&& \int_{\mathcal{X}}\int_{\mathcal{X}} k(\rx, \ry) \PP_1(d\rx)\PP_1(dy) -2\int_{\mathcal{X}}\int_{\mathcal{X}} k(\rx, \ry) \PP_1(d\rx)\PP_2(dy)   +\int_{\mathcal{X}}\int_{\mathcal{X}} k(\rx, \ry) \PP_2(d\rx)\PP_2(dy) \nonumber\\
	= && \EE_{\rx, \ry\sim \PP_1}(k(\rx, \ry)) - 2\EE_{x\sim \PP_1, y\sim \PP_2}(k(\rx, \ry)) + \EE_{\rx, \ry\sim \PP_2}(k(\rx, \ry)),
\end{eqnarray}
which provides a closed form expression up to calculation of expectations.

\medskip
\textbf{Optimization of QCBMs with MMD loss.} The goal of QCBMs is finding an estimator minimizing the loss function $\MMD^2(\PP_{\btheta}||\QQ)$ in Eq.~(\ref{append:eqn:mmd_loss}), where $\PP_{\btheta}$ is defined in Eq.~(\ref{eqn:QCBM}). The optimization is completed by the   gradient based descent optimizer. The updating rule satisfies  $\btheta^{(t+1)}=\btheta^{(t)} - \eta \nabla_{\btheta} \MMD^2(\PP_{\btheta}||\QQ)$ and $\eta$ is the learning rate. Concretely, the partial derivative of the $j$-th entry satisfies
\begin{eqnarray}\label{append:eqn:grad_mmd_qcbm}
	&& \frac{\partial \MMD^2(\PP_{\btheta}||\QQ)}{\partial \btheta_j} \nonumber\\
	= && \frac{\partial  \EE_{\rx, \ry\sim \PP_{\btheta}}(k(\rx, \ry)) - 2\EE_{x\sim \PP_{\btheta}, y\sim \QQ}(k(\rx, \ry)) + \EE_{\rx, \ry\sim \QQ}(k(\rx, \ry))}{\partial \btheta_j} \nonumber\\
	= && \sum_{\rx,\ry} k(\rx, \ry)\left(\PP_{\btheta}(\ry)\frac{\partial  \PP_{\btheta}(\rx)}{\partial \btheta_j} + \PP_{\btheta}(\rx)\frac{\partial  \PP_{\btheta}(\ry)}{\partial \btheta_j} \right) - 2\sum_{\rx, \ry}k(\rx, \ry) \frac{\partial \PP_{\btheta}(\rx)}{\partial \btheta_j} \QQ(\ry) \nonumber\\
	= && \sum_{\rx,\ry} k(\rx, \ry)\left(\PP_{\btheta}(\ry)\left(\PP_{\btheta+\frac{\pi}{2}\bm{e}_j}(\rx) - \PP_{\btheta - \frac{\pi}{2}\bm{e}_j}(\rx) \right)   + \PP_{\btheta}(\rx)\left(\PP_{\btheta+\frac{\pi}{2}\bm{e}_j}(\ry) - \PP_{\btheta - \frac{\pi}{2}\bm{e}_j}(\ry) \right) \right) \nonumber\\
	&& - 2\sum_{\rx, \ry}k(\rx, \ry) \left(\PP_{\btheta+\frac{\pi}{2}\bm{e}_j}(\rx) - \PP_{\btheta - \frac{\pi}{2}\bm{e}_j}(\rx) \right) \QQ(\ry) \nonumber\\
	= &&  \EE_{\rx\sim \PP_{\btheta+\frac{\pi}{2}\bm{e}_j}, \ry\sim \PP_{\btheta}}(k(\rx, \ry)) - \EE_{\rx\sim \PP_{\btheta-\frac{\pi}{2}\bm{e}_j}, \ry\sim \PP_{\btheta}}(k(\rx, \ry)) + \EE_{\rx\sim \PP_{\btheta}, \ry\sim \PP_{\btheta+\frac{\pi}{2}\bm{e}_j}}(k(\rx, \ry)) - \EE_{\rx\sim \PP_{\btheta}, \ry\sim \PP_{\btheta-\frac{\pi}{2}\bm{e}_j}}(k(\rx, \ry))\nonumber\\
	&& - 2\EE_{\rx\sim \PP_{\btheta+\frac{\pi}{2}\bm{e}_j}, \ry\sim \QQ}(k(\rx, \ry)) +  2\EE_{\rx\sim \PP_{\btheta-\frac{\pi}{2}\bm{e}_j}, \ry\sim \QQ}(k(\rx, \ry)) , 
\end{eqnarray} 
where the last second equality employs  the parameter shift rule \cite{schuld2019evaluating} to calculate the partial derivative ${\partial  \PP_{\btheta}(\ry)}/{\partial \btheta_j}$ and ${\partial  \PP_{\btheta}(\rx)}/{\partial \btheta_j}$. According to Lemma \ref{lem:equi-QCBM-QGAN}, the six expectation terms in Eq.~(\ref{append:eqn:grad_mmd_qcbm}) can be analytically and efficiently calculated when the $k(\cdot, \cdot)$ is quantum. In the case of classical kernels,  the six expectation terms in Eq.~(\ref{append:eqn:grad_mmd_qcbm})  are estimated by the sample mean.

\medskip
\textbf{Optimization of QGANs with MMD loss.} We next derive the gradients of QGANs with respect to the $\ell$-th entry. Since the evaluation of expectation is runtime expensive when $\QQ$ is continuous, QGANs employ an unbiased estimator of the MMD loss in Eq.~(\ref{append:eqn:unbia-MMD}) to update $\btheta$. The updating rule at the $t$-th iteration is $\btheta^{(t+1)}=\btheta^{(t)}-\eta\nabla_{\btheta} \MMD^2_U(\PP_{\btheta}||\QQ)$.   According to the chain rule, we have
\begin{eqnarray}
	&& \frac{\partial \MMD_U^2(\PP_{\btheta}||\QQ)}{\partial \btheta_\ell} \nonumber\\
	= && \frac{\partial  \frac{1}{n(n-1)}\sum_{i\neq i'}^n k(G_{\btheta}(\bm{z}^{(i)}), G_{\btheta}(\bm{z}^{(i')}))  - \frac{2}{nm}\sum_{i,j} k(G_{\btheta}(\bm{z}^{(i)}), \ry^{(j)})}{\partial \btheta_\ell} \nonumber\\
	= && \frac{1}{n(n-1)}\sum_{i\neq i'}^n  \frac{\partial  k(G_{\btheta}(\bm{z}^{(i)}), G_{\btheta}(\bm{z}^{(i')}))}{\partial G_{\btheta}(\bm{z}^{(i)})}\frac{\partial G_{\btheta}(\bm{z}^{(i)})}{\partial \btheta_\ell}+\frac{\partial  k(G_{\btheta}(\bm{z}^{(i)}), G_{\btheta}(\bm{z}^{(i')}))}{\partial G_{\btheta}(\bm{z}^{(i')})}\frac{\partial G_{\btheta}(\bm{z}^{(i')})}{\partial \btheta_\ell} -\\ &&\frac{2}{nm}\sum_{i,j} \frac{\partial   k(G_{\btheta}(\bm{z}^{(i)}), \ry^{(j)})}{\partial G_{\btheta}(\bm{z}^{(i)})} \frac{\partial G_{\btheta}(\bm{z}^{(i)})}{\partial \btheta_\ell}
\end{eqnarray}
where the first equality uses $ \partial \frac{1}{m(m-1)}\sum_{j\neq j'}^m k(\ry^{(j)}, \ry^{(j')})/\partial \btheta_\ell = 0$, each derivative ${\partial  k(G_{\btheta}(\bm{z}^{(i)}), G_{\btheta}(\bm{z}^{(i')}))}/{\partial G_{\btheta}(\bm{z}^{(i)})}$ can be easily computed for standard kernels, and the derivative ${\partial G_{\btheta}(\bm{z}^{(i)})}/{\partial \btheta_\ell}$ for $\forall i\in n, j\in [m]$ can be computed via the parameter shift rule. Therefore, the gradients of QGANs with MMD loss can be achieved.

\section{The comparison between QCBMs and GLMs for estimating discrete distributions}\label{appd:coro:DNN-QCBM}
In this section, we emphasize a central question in QGLMs, i.e., when both variational quantum circuits and neural networks are used to implement $\PP_{\btheta}$, which one can attain a lower $\inf_{\btheta\in\bm{\Theta}}\MMD^2 (\PP_{\bm{\theta}}||\QQ)$. The importance of this issue comes from  Eq.~(\ref{eqn:gene_error_bound-exac}), where the generalization error bound becomes meaningful when $\inf_{\btheta\in\bm{\Theta}}\MMD^2 (\PP_{\bm{\theta}}||\QQ)$ is small. In this respect, it is necessary to understand whether QCBMs allow a lower $\inf_{\btheta\in\bm{\Theta}}\MMD^2 (\PP_{\bm{\theta}}||\QQ)$ over (classical) GLMs.

In what follows, we analyze when QCBMs promise a lower $\inf_{\btheta\in\bm{\Theta}}\MMD^2 (\PP_{\bm{\theta}}||\QQ)$ over a typical GLM---restricted Boltzmann machine (RBM) \cite{hinton2012practical}. To be more specific, consider that both QCBMs and RBMs are universal approximators with an exponential number of trainable parameters \cite{biamonte2021universal,le2008representational} with $\inf_{\btheta\in\bm{\Theta}}\MMD^2 (\PP_{\bm{\theta}}||\QQ)\rightarrow 0$, we focus on the more practical scenario in which the number of parameters polynomially scales with the feature dimension, the qudit count, and the number of visible neurons. Denote the space of the parameterized distributions formed by QCBMs (or RBM) as $\PP^{\text{QCBM}}_{\bm{\Theta}}$ (or $\PP^{\text{RBM}}_{\bm{\Theta}}$). The superiority of QCBMs can be identified by showing $\inf_{\btheta\in\bm{\Theta}}\MMD^2 (\PP_{\bm{\theta}}^{\text{QCBM}}||\QQ)\leq \inf_{\btheta\in\bm{\Theta}}\MMD^2 (\PP_{\bm{\theta}}^{\text{RBM}}||\QQ)$. This amounts to finding a distribution $\QQ$ satisfying 
\begin{equation}
	\left(\QQ\in \PP^{\text{QCBM}}_{\bm{\Theta}}\right) \wedge \left(\QQ\notin \PP^{\text{RBM}}_{\bm{\Theta}}\right). 
\end{equation}
According to the results in \cite{gao2017efficient1}, there is a large class of quantum states meeting the above requirement. Representative examples include projected entangled pair states and ground states of $k$-local Hamiltonians.

\section{Proof of Theorem \ref{thm:Gene-QGAN} (generalization of QGANs) }\label{appendix:gene_QGAN}

The proof of Theorem \ref{thm:Gene-QGAN} utilizes the following three lemmas. For clearness, we defer the proof of  Lemmas \ref{lem:cov-T2} and \ref{lem:cov-T3}  to SM  \ref{appendix:lemcov-T2} and \ref{appendix:lemcov-T3}, respectively.  
\begin{lem}[McDiarmids inequality, \cite{mendelson2003few}]\label{lem:mc-ineql}
	Let $f : \mathcal{X}_1\times \mathcal{X}_2 \times ... \times \mathcal{X}_N\rightarrow \mathbb{R}$ and assume there exists $c_1,...,c_2\geq 0$ such that, for all $k \in\{1,...,N\}$, we have
	\begin{equation}
\sup_{\rx_1,...,\rx_k,\tilde{\rx}_k,...\rx_N} |f(\rx_1,...,\rx_k,...,\rx_N)-f(\rx_1,...,\tilde{\rx}_k,...,\rx_N)|\leq c_k.
	\end{equation}
Then for all $\epsilon\geq 0$ and independent random variables $\xi_1,...\xi_N$ in $\mathcal{X}$, 
\begin{equation}
	\Pr(|f(\xi_1,...\xi_N)-\EE(f(\xi_1,...\xi_N))|\geq \epsilon)\leq \exp\left(\frac{-2\epsilon^2}{\sum_{n=1}^Nc_n^2}\right).
\end{equation} 
\end{lem}

\begin{lem}\label{lem:cv-QGAN}\label{lem:cov-T2}
Following the notations in Theorem \ref{thm:Gene-QGAN},  define $\mathcal{G}=\{k(G_{\bm{\theta}}(\cdot), G_{\bm{\theta}}(\cdot))|\bm{\theta}\in \bm{\Theta}\}$ and $\mathcal{G}_+=\{k(G_{\bm{\theta}}(\bm{z}), G_{\bm{\theta}}(\cdot))|\bm{\theta}\in \bm{\Theta}, \bm{z}\in \mathcal{Z}\}$. Given the set $\mathcal{S}=\{\bm{z}^{(i)}\}_{i=1}^n$, we have
\begin{eqnarray}
	&& \EE\large(\sup_{\btheta\in \bm{\Theta}}\large| \EE_{\bm{z},\bm{z}'}(k(G_{\btheta}(\bm{z}),G_{\btheta}(\bm{z}')))- \frac{1}{n(n-1)}\sum_{i\neq i'}k(G_{\btheta}(\bm{z}^{(i)}), G_{\btheta}(\bm{z}^{(i')}))\large|\large)\nonumber\\
	 \leq && \frac{8}{n-1} + \frac{24\sqrt{d^{2k}(N_{ge}+N_{gt})}}{n-1}\left(1+  N\ln (441d C_3^2(n-1)N_{ge}N_{gt}) \right).
\end{eqnarray}
\end{lem}

\begin{lem}\label{lem:cov-T3}
Following the notations in Theorem \ref{thm:Gene-QGAN},  define $\mathcal{W}=\{k(G_{\bm{\theta}}(\cdot), \cdot)|\bm{\theta}\in \bm{\Theta}\}$ and $\mathcal{W}_+=\{k(G_{\bm{\theta}}(\cdot), \ry)|\bm{\theta}\in \bm{\Theta}, \bm{y}\in \mathcal{Y}\}$. Given the set $\mathcal{S}=\{\bm{z}^{(i)}\}_{i=1}^n$ and the set $\{\ry^{(j)}\}_{j=1}^m$, we have
\begin{eqnarray}
	&& \EE\left(\sup_{\btheta\in \bm{\Theta}}\Big|\EE_{\bm{z},\ry}(k(G_{\btheta}(\bm{z}),\ry))- \frac{1}{mn}\sum_{i\in[n],j\in[m]}k(G_{\btheta}(\bm{z}^{(i)}), \ry^{(j)})\Big| \right) \nonumber\\
	\leq &&  \frac{8}{n} + \frac{24\sqrt{d^{2k}(N_{ge}+N_{gt})}}{n}\left(1+  N\ln (441dC_3^2nN_{ge}N_{gt}) \right).
\end{eqnarray}
\end{lem}

\medskip
We are now ready to prove Theorem \ref{thm:Gene-QGAN}.
\begin{proof}[Proof of Theorem \ref{thm:Gene-QGAN}]

Let $\mathcal{E}(\btheta)=\MMD_U^2(\PP_{\btheta}^n||\QQ^m)$ and $\mathcal{T}(\btheta)=\MMD^2(\PP_{\btheta}||\QQ)$. Note that the generalization error equals to
\begin{equation}
	\mathfrak{R}^C= \mathcal{T}(\hat{\btheta}^{(n,m)})-\mathcal{T}( \btheta^*).
\end{equation}
In the remainder of the proof, when no confusion occurs, we abbreviate $\hat{\btheta}^{(n,m)}$ as $\hat{\btheta}$ for clearness. The above equation is upper bounded by
\begin{equation}
	\mathfrak{R}^C = \mathcal{E}(\hat{\btheta})-\mathcal{E}( \hat{\btheta}) +\mathcal{T}(\hat{\btheta})-\mathcal{T}( \btheta^*) \leq \mathcal{E}(\btheta^*)-\mathcal{E}( \hat{\btheta}) +\mathcal{T}(\hat{\btheta})-\mathcal{T}( \btheta^*)\leq |\mathcal{T}( \hat{\btheta}) -\mathcal{E}( \hat{\btheta})| + |\mathcal{T}(\btheta^*) -\mathcal{E}(\btheta^*)|,
\end{equation} 
where the first inequality employs the definition of $\hat{\btheta}$ with $\mathcal{E}(\btheta^*)\geq \mathcal{E}(\hat{\btheta})$ and the second inequality uses the property of absolute value function. In the following, we derive the probability for $\sup_{\btheta \in \Theta} |\mathcal{T}(\btheta) -\mathcal{E}(\btheta)| < \epsilon$,  which in turn can achieve the upper bound of $\mathfrak{R}^C$. 

\medskip 
According to the explicit form of the MMD loss, $\sup_{\btheta\in \bm{\Theta}}|\mathcal{E}(\btheta)  - \mathcal{T}(\btheta)|$ satisfies
\begin{eqnarray}\label{eqn:proof_thm3_0}
	&& \sup_{\btheta\in \bm{\Theta}}\left|\MMD_U^2(\PP_{\btheta}^n||\QQ^m) - \MMD^2(\PP_{\btheta}||\QQ)\right| \nonumber\\
	= && \sup_{\btheta\in \bm{\Theta}}\Big|\EE_{\bm{z},\bm{z}'}(k(G_{\btheta}(\bm{z}),G_{\btheta}(\bm{z}'))) - 2\EE_{\bm{z}, y}(k(G_{\btheta}(\bm{z}),\ry)) + \EE_{\ry,\ry'}(k(\ry,\ry')) - \frac{1}{n(n-1)}\sum_{i\neq i'}k(G_{\btheta}(\bm{z}^{(i)}), G_{\btheta}(\bm{z}^{(i')}))\nonumber\\
	&& + \frac{2}{mn}\sum_{i\in[n],j\in[m]}k(G_{\btheta}(\bm{z}^{(i)}), \ry^{(j)}) - \frac{1}{m(m-1)}\sum_{j\neq j'}k(\ry^{(j)}, \ry^{(j')})\Big|\nonumber\\
	\leq && \underbrace{\sup_{\btheta\in \bm{\Theta}}\Big|\EE_{\ry,\ry'}(k(\ry,\ry')) - \frac{1}{m(m-1)}\sum_{j\neq j'}k(\ry^{(j)}, \ry^{(j')})\Big|}_{T1} + \underbrace{\sup_{\btheta\in \bm{\Theta}}\Big| \EE_{\bm{z},\bm{z}'}(k(G_{\btheta}(\bm{z}),G_{\btheta}(\bm{z}')))- \frac{1}{n(n-1)}\sum_{i\neq i'}k(G_{\btheta}(\bm{z}^{(i)}), G_{\btheta}(\bm{z}^{(i')}))\Big|}_{T2} \nonumber\\
	&& + 2\underbrace{\sup_{\btheta\in \bm{\Theta}}\Big|\EE_{z, y}(k(G_{\btheta}(\bm{z}),\ry))- \frac{1}{mn}\sum_{i\in[n],j\in[m]}k(G_{\btheta}(\bm{z}^{(i)}), \ry^{(j)})\Big|}_{T3}
\end{eqnarray} 
where the inequality comes from the Jensen inequality.

\medskip
We next separately derive the upper bounds of the terms $T1$, $T2$, and $T3$ in Eq.~(\ref{eqn:proof_thm3_0}).
	
\medskip	
\noindent\textit{\underline{Upper bound of $T1$.}} The upper bound of $T1$ only depends on the examples sampled from the target distribution $\QQ$, which is independent of $\btheta\in \bm{\Theta}$. With this regard, T1 can be taken out of the supremum and we apply the concentration inequality in  Lemma \ref{lem:mc-ineql} to derive the upper bound of $|\EE_{\ry,\ry'}(k(\ry,\ry')) - \frac{1}{m(m-1)}\sum_{j\neq j'}k(\ry^{(j)}, \ry^{(j')})|$. Recall the precondition of employing Lemma \ref{lem:mc-ineql} is finding the upper bound on $f(\cdot)$. Let the function $f(\cdot)$ be $\frac{1}{m(m-1)}\sum_{j\neq j'}k(\ry^{(j)}, \ry^{(j')})$. For each $\ell\in\{1,...,m\}$, the desired upper bound yields 
\begin{eqnarray}
	&& \left| - \frac{1}{m(m-1)}(\sum_{j\neq j', j\neq \ell}k(\ry^{(j)}, \ry^{(j')}) +  \sum_{j'\neq \ell}k(\ry^{(\ell)}, \ry^{(j')}))+ \frac{1}{m(m-1)}(\sum_{j\neq j', j\neq \ell}k(\ry^{(j)}, \ry^{(j')}) +  \sum_{j'\neq \ell}k(\tilde{\ry}^{(\ell)}, \ry^{(j')}))\right|\nonumber\\
	= && \left| \frac{1}{m(m-1)}\sum_{j'\neq \ell}\left(k(\ry^{(\ell)}, \ry^{(j')})) - k(\tilde{\ry}^{(\ell)}, \ry^{(j')}))\right) \right| \nonumber\\
	\leq && \frac{2C_2}{m},
\end{eqnarray} 
where the inequality leverages the assumption that the kernel $k(\cdot, \cdot)$ is upper bounded by $C_2$.

Given this upper bound, we obtain
\begin{eqnarray}\label{eqn:thm3-T1-up}
	&& \Pr(T1\geq \epsilon) \nonumber\\
= && \Pr\left(\left|\EE_{\ry,\ry'}(k(\ry,\ry')) - \frac{1}{m(m-1)}\sum_{j\neq j'}k(\ry^{(j)}, \ry^{(j')})\right|\geq \epsilon \right) \nonumber\\
\leq && \exp\left(-\frac{\epsilon^2}{8C_2^2}m \right)=\delta_{T1},
\end{eqnarray} 
 where the inequality exploits the results in  Lemma \ref{lem:mc-ineql}. 

\noindent\textit{\underline{Upper bound of $T2$.}} 
We next use the concentration inequality to quantify the upper bound of $T2$. The derivation is similar to that of $T1$. In particular, supported by the results of Lemma \ref{lem:mc-ineql}, we have
\begin{equation}\label{eqn:thm2-5}
	\Pr(|T2-\EE(T2) |\geq \epsilon )\leq \exp\left(-\frac{\epsilon^2}{8C_2^2}n \right)=\delta_{T2}.
\end{equation}
Suppose that $\EE(T2)\leq \epsilon_1$, an immediate observation is that
\begin{equation}\label{eqn:thm3-T2-up}
	\Pr(T2 \geq \epsilon_1 + \epsilon )\leq \exp\left(-\frac{\epsilon^2}{8C_2^2}n \right)=\delta_{T2}.
\end{equation}
In other words, the derivation of the upper bound of T2 amounts to analyzing the upper bound $\epsilon_1$. 

\noindent\textit{\underline{Upper bound of $T3$.}} Following the same routine with the derivation of the upper bound of T2, we obtain
\begin{equation}
	\Pr(|T3 -\EE(T3)| \geq  \epsilon )\leq \exp\left(-\frac{\epsilon^2}{8C_2^2}\frac{nm}{n+m} \right)=\delta_{T3}.
\end{equation}
Suppose that $\EE(T3)\leq \epsilon_2$. The above result hints that
\begin{equation}\label{eqn:thm3-T3-up}
	\Pr(T3 \geq \epsilon_2 + \epsilon )\leq \exp\left(-\frac{\epsilon^2}{8C_2^2}\frac{nm}{n+m} \right)=\delta_{T3}.
\end{equation}

Summing up Eqs.~(\ref{eqn:thm3-T1-up}), (\ref{eqn:thm3-T2-up}), and (\ref{eqn:thm3-T3-up}), the union bound gives 
\begin{eqnarray}
	&& \Pr\left(\sup_{\btheta\in \bm{\Theta}}|\mathcal{E}(\btheta)  - \mathcal{T}(\btheta)|\geq \epsilon_1 + 2\epsilon_2 + 4\epsilon \right)  
	\leq   \delta_{T1} + \delta_{T2}  +\delta_{T3}   \nonumber\\
\Rightarrow && \Pr\left(\sup_{\btheta\in \bm{\Theta}}|\mathcal{E}(\btheta)  - \mathcal{T}(\btheta)|\geq \epsilon_1 + 2\epsilon_2 + 4\epsilon \right)  
	\leq   3\delta_{T3}  \nonumber\\	
\Rightarrow && 	\Pr\left(|\mathcal{E}(\hat{\btheta})  - \mathcal{T}(\hat{\btheta})|\geq \epsilon_1 + 2\epsilon_2 + 4\epsilon \right)  
	\leq   3\delta_{T3}. 
\end{eqnarray}
This yields that with probability at least $1- 3 \delta_{T3}$, 
\begin{equation}
	2 (\epsilon_1 + 2\epsilon_2 + 4\epsilon) \geq |\mathcal{E}(\hat{\btheta})  - \mathcal{T}(\hat{\btheta})| + |\mathcal{E}(\btheta^*)  - \mathcal{T}( \btheta^*)| \geq |\mathcal{E}(\hat{\btheta})  - \mathcal{T}(\hat{\btheta}) - \mathcal{E}(\btheta^*)  +  \mathcal{T}(\btheta^*)| \geq |\mathcal{T}(\hat{\btheta})  -\mathcal{T}(\btheta^*)|=\mathfrak{R}^C. 
\end{equation}
According to the explicit forms of $\epsilon_1$ and $\epsilon_2$ achieved in Lemmas \ref{lem:cov-T2} and \ref{lem:cov-T3}, with probability $1-3\delta_{T3}$, the generalization error of QGANs is upper bounded by
\begin{eqnarray}
	&&\mathfrak{R}^C \nonumber\\
	\leq &&  8\epsilon + 2 \left( \frac{8}{n-1} + \frac{24\sqrt{d^{2k}(N_{ge}+N_{gt})}}{n-1}\left(1+  N\ln (1764C_3^2(n-1)N_{ge}N_{gt}) \right) \right)  \nonumber\\
	&& + 4\left(\frac{8}{n} + \frac{24\sqrt{d^{2k}(N_{ge}+N_{gt})}}{n}\left(1+  N\ln (1764C_3^2nN_{ge}N_{gt}) \right) \right) \nonumber\\
	\leq && 8\sqrt{ \frac{8C_2^2(n+m)}{nm}\ln(\frac{1}{3\delta_{T3}}) }  + 6 \left(\frac{8}{n-1} + \frac{24 \sqrt{d^{2k}(N_{gt}+N_{ge})}}{n-1} \left(1+  N\ln (441dC_3^2nN_{ge}N_{gt}) \right)\right).
\end{eqnarray}
where the last inequality uses the relation between $\delta_{T3}$ and $\epsilon$ in Eq.~(\ref{eqn:thm3-T3-up}) with $\epsilon=\sqrt{8C_2^2(n+m)\ln(1/\delta_{T3})/(nm)}$, $1/n<1/(n-1)$, and $n-1<n$.

\end{proof}
			
\subsection{Proof of Lemma \ref{lem:cov-T2}}\label{appendix:lemcov-T2}

Recall Lemma \ref{lem:cov-T2} aims to derive the upper bound of $\EE(T2)$ in Eq.~(\ref{eqn:proof_thm3_0}). In Ref.~\cite{dziugaite2015training}, the authors utilize a statistical measure named Rademacher complexity to quantify these two terms. Different from the classical counterpart, here we adopt an another statistical measure, i.e., covering number, to derive the upper bound of $\EE(T2)$. This measure allows us to identify how $\EE(T2)$  scales with the qudit count $N$ and the architecture of the employed Ansatz such as the trainable parameters $N_{gt}$ and the types of the quantum gates. For self-consistency, we provide the formal definition of covering number and Rademacher as follows.
\begin{definition}[Covering number, \cite{MohriRostamizadehTalwalkar18}]\label{def:cov-num}
	The covering number $\mathcal{N}(\mathcal{U}, \epsilon, \|\cdot \|)$ denotes the least cardinality of any subset $\mathcal{V} \subset \mathcal{U}$ that covers $\mathcal{U}$ at scale $\epsilon$ with a norm $ \|\cdot \|$, i.e.,
	\begin{equation}
		\sup_{A\in \mathcal{U}} \min_{B\in \mathcal{V}} \|A - B \|\leq \epsilon.
	\end{equation}
\end{definition} 
\begin{definition}[Rademacher, \cite{MohriRostamizadehTalwalkar18}]\label{def:radem}
	Let $\mu$ be a probability measure on $\mathcal{X}$, and let $\mathcal{F}$ be a class of uniformly bounded functions on $\mathcal{X}$. Then the Rademacher complexity of $\mathcal{F}$ is 
	\begin{equation}
		\mathfrak{R}_{n}(\mathcal{F})=\mathbb{E}_{\mu}\mathbb{E}_{\sigma_1,...,\sigma_n}\left(\frac{1}{\sqrt{n}}\sup_{f\in\mathcal{F}}\left|\sum_{i=1}^n \sigma_if(\rx^{(i)})\right| \right),
	\end{equation}
	where $\bm{\sigma}=(\sigma_1,...,\sigma_n)$ is a sequence of independent Rademacher variables taking values in $\{-1, 1\}$ and each with probability $1/2$, and $\rx_1,...,\rx_n \in\mathcal{X}$ are independent, $\mu$-distributed random variables.
\end{definition}
Intuitively, the covering number concerns the minimum number of spherical balls with radius $\varepsilon$ that occupies the whole space; the Rademacher complexity measures the ability of functions from $\mathcal{F}$ to fit random noise. The relation between Rademacher complexity and covering number is established by the following Dudley entropy integral bound.
\begin{lem}[Adapted from \cite{dziugaite2015training, dudley1967sizes}]\label{lem:rad-cov-dudley} 
Let $\mathcal{F}=\{f: \mathcal{X}\times \mathcal{X} \rightarrow \mathbb{R} \}$ and $\mathcal{F}_+=\{h=f(\rx,\cdot):f\in \mathcal{F}, \rx\in\mathcal{X} \}$ and $\mathcal{F}_+\subset B(L_{\infty}(\mathcal{X}))$. Given the set $\mathcal{S}=\{\rx^{(1)},...,\rx^{(n)}\}\in\mathcal{X}$, denote the   Rademacher complexity of $\mathcal{F}_+$ as $\mathfrak{R}_n(\mathcal{F}_+)$, it satisfies 
\begin{equation}
	\mathfrak{R}_n(\mathcal{F}_+) \leq \inf_{\alpha>0}\left(4\alpha + \frac{12}{\sqrt{n}}\int_{\alpha}^{1}\sqrt{\ln(\mathcal{N}((\mathcal{F}_+)_{|\mathcal{S}}, \epsilon, \|\cdot \|_2)}  d\epsilon\right), 
\end{equation}
where $(\mathcal{F}_+)_{|\mathcal{S}}=\{[f(\bm{x}, \bm{x}^{(i)})]_{i=1:n}:f\in \mathcal{F}, \rx\in\mathcal{X}\}$ denotes the set of vectors formed by the hypothesis with $\mathcal{S}$.    
\end{lem}

Ref.~\cite{dziugaite2015training} hinges on the term $\mathbb{E}(T2)$ with the Rademacher complexity, as stated in the following lemma.
\begin{lem}[Adapted from Lemma 1, \cite{dziugaite2015training}]\label{lem:T2-T3-rademacher}
Following notations in Theorem \ref{thm:Gene-QGAN} and Lemma \ref{lem:rad-cov-dudley},  define $\mathcal{G}=\{k(G_{\bm{\theta}}(\cdot), G_{\bm{\theta}}(\cdot))|\bm{\theta}\in \bm{\Theta}\}$ and $\mathcal{G}_+=\{k(G_{\bm{\theta}}(\bm{z}), G_{\bm{\theta}}(\cdot))|\bm{\theta}\in \bm{\Theta}, \bm{z}\in \mathcal{Z}\}$. Given the set $\mathcal{S}=\{\bm{z}^{(i)}\}_{i=1}^n$, we have  
\[\EE(T2)\leq \frac{2}{\sqrt{n-1}}\mathfrak{R}_{n-1}(\mathcal{G}_+),\]
where $\mathfrak{R}_{n-1}(\mathcal{G}_+)$ refers to the Rademacher's complexity of $\mathcal{G}_+$.  
\end{lem}
In conjunction with the above two lemmas, the term $\mathbb{E}(T2)$ is upper bounded by the covering number of $\mathcal{G}_+$. As such, the proof of Lemma \ref{lem:cov-T2} utilizes the following three lemmas, which are used to formalize the relation of covering number of two metric spaces, quantify the covering number of variational quantum circuits, and evaluate the covering number of  the space living in $N$-qudit quantum states, respectively.
\begin{lem}[Lemma 5, \cite{Barthel2018fundamental}]\label{lem:cov-num-two-space}
	Let $(\mathcal{H}_1, d_1)$ and $(\mathcal{H}_2, d_2)$ be metric spaces and $f:\mathcal{H}_1 \rightarrow  \mathcal{H}_2 $ be bi-Lipschitz such that 
	\begin{equation}\label{eqn:lem-cov-num-two-space-1}
		d_2(f(\bm{x}), f(\bm{y})) \leq K d_1(\bm{x}, \bm{y}), ~\forall \bm{x}, \bm{y} \in \mathcal{H}_1,
	\end{equation}
	and 
	\begin{equation}\label{eqn:lem-cov-num-two-space-2}
		d_2(f(\bm{x}), f(\bm{y})) \geq k d_1(\bm{x}, \bm{y}), ~\forall \bm{x}, \bm{y} \in \mathcal{H}_1~\text{with}~ d_1(\bm{x}, \bm{y}) \leq r.
	\end{equation}
Then their covering numbers obey
\begin{equation}\label{eqn:lem-cov-num-two-space-3}
	\mathcal{N}(\mathcal{H}_1, 2\epsilon /k, d_1 ) \leq \mathcal{N}(\mathcal{H}_2,  \epsilon, d_2   ) \leq  \mathcal{N}(\mathcal{H}_1, \epsilon /K, d_1 ),
\end{equation}
where the left inequality requires $\epsilon \leq kr/2$.
\end{lem}

\begin{lem}[Lemma 2, \cite{du2021efficient}]\label{lem:cov-cirt-ansa} 
Define the operator group as
\begin{equation}\label{eqn:def-opt-group}
	\mathcal{H}_{circ}:=\left\{\hat{U}(\bm{\theta})\Pi_j\hat{U}(\bm{\theta})^{\dagger}| \bm{\theta}\in \Theta \right\}.
\end{equation}
Suppose that the employed encoding Ansatz $\hat{U}(\bm{\theta})$ containing in total $N_g$ gates, each gate $\hat{u}_i(\bm{\theta})$ acting on at most $k$ qudits, and $N_{gt}\leq N_g$ gates in $\hat{U}(\bm{\theta})$ are trainable. The $\epsilon$-covering number for the operator  group $\mathcal{H}_{circ}$  with respect to the operator-norm distance obeys
\begin{equation}
  \mathcal{N}(\mathcal{H}_{circ}, \epsilon , \|\cdot\|) \leq \left(\frac{7 N_{gt}  \|\Pi_j\| }{\epsilon} \right)^{d^{2k}N_{gt}},
\end{equation}
where $\|\Pi_j\|$ denotes the operator norm of $\Pi_j$.
\end{lem}

\begin{lem}\label{append:lem_cover_input_state}
Define the input state group as 
	$\mathcal{B}:=\left\{\rho_{\bm{z}}:=\hat{U}(\bm{z})^{\dagger}(\ket{\bm{0}}\bra{\bm{0}})\hat{U}(\bm{z}) \Big| \bm{z}\in \mathcal{Z} \right\}$. Suppose that the employed quantum circuit $\hat{U}(\bm{z})$ containing in total $N_{ge}$ parameterized gates to load $\bm{z}$ and each gate $\hat{u}_i(\bm{z})$ acting on at most $k$ qudits. The $\epsilon$-covering number for $\mathcal{B}$  with respect to the operator-norm distance obeys
\begin{equation}\label{append:eqn:encode-cover}
  \mathcal{N}(\mathcal{B}, \epsilon , \|\cdot\|) \leq \left(\frac{7 N_{ge} }{\epsilon} \right)^{d^{2k}N_{ge}}.
\end{equation}	
\end{lem}
\begin{proof}[Proof of Lemma \ref{append:lem_cover_input_state}]
The proof is identical to that presented in Lemma 2 of Ref.~\cite{du2021efficient}. 
\end{proof}

\medskip
We are now ready to prove Lemma \ref{lem:cov-T2}.
\begin{proof}[Proof of Lemma \ref{lem:cov-T2}] Recall the aim of Lemma \ref{lem:cov-T2} is to obtain the upper bound of $\EE(T2)$. In conjunction with Lemmas  \ref{lem:rad-cov-dudley} and \ref{lem:T2-T3-rademacher}, we obtain
\begin{equation}\label{eqn:proof-lem4-0-0}
	\EE(T2)\leq \EE\frac{2}{\sqrt{n-1}} \inf_{\alpha>0} \left(4\alpha + \frac{12}{\sqrt{n-1}}\int_{\alpha}^{1}\sqrt{\ln(\mathcal{N}((\mathcal{G}_+)_{|\mathcal{S}}, \epsilon, \|\cdot \|_2)}  d\epsilon\right),
\end{equation}
where $\mathcal{G}_+=\{k(G_{\bm{\theta}}(\bm{z}), G_{\bm{\theta}}(\cdot))|\bm{\theta}\in \bm{\Theta}, \bm{z}\in \mathcal{Z}\}$ and $\mathcal{S}$ denotes the set $\{\bm{z}^{(1)},...,\bm{z}^{(n-1)}\}$ sampled from the prior distribution $\PP_{\mathcal{Z}}$, and $(\mathcal{G}_+)_{|\mathcal{S}}=\{[k(G_{\bm{\theta}}(\bm{z}), G_{\bm{\theta}}(\bm{z}^{(i)})]_{i=1:n-1}: \btheta\in \bm{\Theta}, \bm{z}\in\mathcal{Z}\}$ denotes the set of vectors formed by the hypothesis with $\mathcal{S}$. In other words, the upper bound of $\EE(T2)$ is quantified by the covering number $\mathcal{N}((\mathcal{G}_+)_{|\mathcal{S}}, \epsilon, \|\cdot \|_2)$.

We next follow the definition of covering number to quantify how $\mathcal{N}((\mathcal{G}_+)_{|\mathcal{S}}, \epsilon, \|\cdot \|_2)$ depends on the structure information of the employed Ansatz and the input quantum states. Denote $\mathcal{Q}_{\epsilon_1}$ as an $\epsilon_1$-covering of the set $\mathcal{Q}_1=\{G_{\btheta}(\bm{z})|\btheta\in \bm{\Theta}\}$ and $\mathcal{Q}_{\epsilon_3}$ as an $\epsilon_3$-covering of the set $\mathcal{Q}=\{G_{\btheta}(\bm{z})|\bm{z}\in \bm{\mathcal{Z}}\}$. Then, the covering number $\mathcal{N}((\mathcal{G}_+)_{|\mathcal{S}}, \epsilon, \|\cdot \|_2)$ can be upper bounded by $\mathcal{N}((\mathcal{Q}_1)_{|\mathcal{S}}, \epsilon_1, \|\cdot \|_2)$ and $\mathcal{N}((\mathcal{Q}_3)_{|\mathcal{S}}, \epsilon_3, \|\cdot \|_3)$.  Mathematically, according to the explicit expression of $(\mathcal{G}_+)_{|\mathcal{S}}$,  we have for any $(\bm{\theta},\bm{z})$ and $(\bm{\theta}',\bm{z}')$  
\begin{eqnarray}\label{append:eqn:lem4-1}
	 && \Big\| [k(G_{\bm{\theta}}(\bm{z}), G_{\bm{\theta}}(\bm{z}^{(i)}))]_{i=1:n-1}  - [k(G_{\bm{\theta}'}(\bm{z}'), G_{\bm{\theta}'}(\bm{z}^{(i)}))]_{i=1:n-1} \Big\|_2   \nonumber\\
	  = &&  \Big\|  [k(G_{\bm{\theta}}(\bm{z}), G_{\bm{\theta}}(\bm{z}^{(i)}))]_{i=1:n-1} - [k(G_{\bm{\theta}'}(\bm{z}), G_{\bm{\theta}}(\bm{z}^{(i)}))]_{i=1:n-1} +  [k(G_{\bm{\theta}'}(\bm{z}), G_{\bm{\theta}}(\bm{z}^{(i)}))]_{i=1:n-1}  \nonumber \\
	  && - [k(G_{\bm{\theta}'}(\bm{z}), G_{\bm{\theta}'}(\bm{z}^{(i)}))]_{i=1:n-1} + [k(G_{\bm{\theta}'}(\bm{z}), G_{\bm{\theta}'}(\bm{z}^{(i)}))]_{i=1:n-1} - [k(G_{\bm{\theta}'}(\bm{z}'), G_{\bm{\theta}'}(\bm{z}^{(i)}))]_{i=1:n-1}\Big\|_2 \nonumber\\
	  \leq &&  \Big\|  [k(G_{\bm{\theta}}(\bm{z}), G_{\bm{\theta}}(\bm{z}^{(i)}))]_{i=1:n-1} - [k(G_{\bm{\theta}'}(\bm{z}), G_{\bm{\theta}}(\bm{z}^{(i)}))]_{i=1:n-1}\Big\|_2 +  \Big\|[k(G_{\bm{\theta}'}(\bm{z}), G_{\bm{\theta}}(\bm{z}^{(i)}))]_{i=1:n-1}  \nonumber \\
	  && - [k(G_{\bm{\theta}'}(\bm{z}), G_{\bm{\theta}'}(\bm{z}^{(i)}))]_{i=1:n-1} \Big\|_2+ \Big\|[k(G_{\bm{\theta}'}(\bm{z}), G_{\bm{\theta}'}(\bm{z}^{(i)}))]_{i=1:n-1} - [k(G_{\bm{\theta}'}(\bm{z}'), G_{\bm{\theta}'}(\bm{z}^{(i)}))]_{i=1:n-1}\Big\|_2 \nonumber\\
	  \leq &&C_3\left(\sqrt{n-1}\big\|G_{\bm{\theta}}(\bm{z})- G_{\bm{\theta}'}(\bm{z})\big\|+\Big\| \big[\big\|G_{\bm{\theta}'}(\bm{z}^{(i)}) - G_{\bm{\theta}}(\bm{z}^{(i)})\big\|\big]_{i=1:n-1}\Big\|_2+\sqrt{n-1}\big\|G_{\bm{\theta}'}(\bm{z})- G_{\bm{\theta}'}(\bm{z}')\big\|\right), 
\end{eqnarray}
where the first inequality uses the triangle inequality and the last inequality exploits $C_3$-Lipschitz property of the kernel.  Following the definition of covering number, the above relationship indicates that if for any $\btheta$ there exists $\btheta'$ such that $\big\|G_{\bm{\theta}}(\bm{z})- G_{\bm{\theta}'}(\bm{z})\big\|\leq \epsilon_1$ holds for every $\bm{z}$, and for any $\bm{z}$ there exists $\bm{z}'$ such that $\big\|G_{\bm{\theta}}(\bm{z})- G_{\bm{\theta}}(\bm{z}')\big\|\leq \epsilon_3$ holds for every $\btheta$, the composition of the covering sets $\mathcal{Q}_{\epsilon_1}$ and $\mathcal{Q}_{\epsilon_3}$ forms the covering set of $(\mathcal{G}_+)_{|\mathcal{S}}$. That is, the covering number of  $(\mathcal{G}_+)_{|\mathcal{S}}$ is upper bounded by
\begin{equation}\label{append:eqn:cov_g+_all}
	\mathcal{N}((\mathcal{G}_+)_{|\mathcal{S}}, C_3\sqrt{n-1}(2\epsilon_1+\epsilon_3), \|\cdot \|_2) \leq \mathcal{N}(\mathcal{Q}_1, \epsilon_1, \|\cdot \|_2) \times \mathcal{N}(\mathcal{Q}_3, \epsilon_3, \|\cdot \|_2).
\end{equation}

In other words, to quantify the $\epsilon$-covering   of $(\mathcal{G}_+)_{|\mathcal{S}}$, it is equivalent to deriving the upper bound of $\mathcal{N}(\mathcal{Q}_1, \epsilon/(3C_3\sqrt{n-1}), \|\cdot \|_2)$ and $\mathcal{N}(\mathcal{Q}_3, \epsilon/(3C_3\sqrt{n-1}), \|\cdot \|_2)$, respectively. We next separately derive these two quantities.

\medskip
\textit{\underline{The upper bound of $\mathcal{N}(\mathcal{Q}_1, \epsilon/(3C_3\sqrt{n-1}), \|\cdot \|_2)$.}}  Let $\mathcal{Q}_4$ be an $\frac{\epsilon}{3C_3 d^{N}\sqrt{n-1} }$-cover of $\mathcal{H}_{circ}$ in Eq.~(\ref{eqn:def-opt-group}). Then, for any $\btheta$, there exists $\btheta'$ such that $\|\hat{U}(\btheta)\Pi_j\hat{U}(\btheta)- \hat{U}(\btheta')\Pi_j\hat{U}(\btheta')\|\leq \frac{\epsilon}{3C_3 d^{N}\sqrt{n-1} }$ for every $j$ with $\hat{U}(\btheta')\Pi_j\hat{U}(\btheta')\in \mathcal{Q}_4$. This leads that for any $\bm{z}$, we have
\begin{eqnarray}
	&& \Big\|  G_{\bm{\theta}}(\bm{z})- G_{\bm{\theta}'}(\bm{z})\Big\|_2 \nonumber\\
	= && \left\| [ \Tr(\hat{U}(\btheta) \Pi_j \hat{U}(\btheta)^{\dagger} \rho_{\bm{z}}) -  \Tr(\hat{U}(\btheta') \Pi_j \hat{U}(\btheta')^{\dagger} \rho_{\bm{z}})]_{j=1:d^N}\right\|_2 \nonumber\\
	\leq  && \left\|  [\LARGE\|  \hat{U}(\btheta) \Pi_j \hat{U}(\btheta)^{\dagger}  - \hat{U}(\btheta') \Pi_j \hat{U}(\btheta')^{\dagger} \LARGE\|]_{j=1:d^N}\right\|_2 \nonumber\\
	\leq && d^{N}  \frac{\epsilon}{3C_3 d^{N}\sqrt{n-1} } \nonumber\\
	=  && \frac{\epsilon}{3C_3\sqrt{n-1}},
\end{eqnarray}
where the first inequality comes from the Cauchy-Schwartz inequality and the second inequality follows the definition of covering number.  

The above observation means that the covering set of $\mathcal{N}(\mathcal{Q}_1, \epsilon/(3C_3\sqrt{n-1}), \|\cdot \|_2)$ is independent with $\bm{z}$ and its covering number is upper bound by $\mathcal{N}(\mathcal{H}_{circ}, \frac{\epsilon}{3C_3 d^{N}\sqrt{n-1} }, \|\cdot\|_2)$. Then, by leveraging the results in Lemma \ref{lem:cov-cirt-ansa}, we obtain 
\begin{equation}\label{eqn:proof-lem4-1}
\mathcal{N}\left(\mathcal{Q}_1, \frac{\epsilon}{3C_3\sqrt{n-1}}, \|\cdot \|_2\right) \leq \mathcal{N}\left(\mathcal{H}_{circ}, \frac{\epsilon}{3C_3 d^{N}\sqrt{n-1} }, \|\cdot\|_2 \right)	\leq  \left(\frac{21C_3 d^{N}\sqrt{n-1} N_{gt}  \|\Pi_j\| }{\epsilon} \right)^{d^{2k}N_{gt}}.\end{equation}

\medskip 
\textit{\underline{The upper bound of $\mathcal{N}(\mathcal{Q}_3, \epsilon/(3C_3\sqrt{n-1}), \|\cdot \|_2)$.}} Let $\mathcal{Q}_5$ be an $\frac{\epsilon}{3C_3 d^{N} \sqrt{n-1}}$-cover of $\mathcal{B}$ in Eq.~(\ref{append:eqn:encode-cover}). Then, for any encoding state $\rho_{\bm{z}}\in \mathcal{B}$, there exists $\rho_{\bm{z}'}\in \mathcal{Q}_{5}$ with $\|\rho_{\bm{z}}- \rho_{\bm{z}'}\|\leq \frac{\epsilon}{3C_3 d^{N} \sqrt{n-1}}$.  By expanding the  term $ \big\|G_{\bm{\theta}'}(\bm{z})- G_{\bm{\theta}'}(\bm{z}')\big\|$, we obtain the following result, i.e., for any $\btheta'$,
\begin{eqnarray}
&&	\big\|G_{\bm{\theta}'}(\bm{z})- G_{\bm{\theta}'}(\bm{z}')\big\|    \nonumber\\
=  &&\big\|[ \Tr(\Pi_j \hat{U}(\btheta') \rho_{\bm{z}} \hat{U}(\btheta')^{\dagger} ) -  \Tr(\Pi_j \hat{U}(\btheta') \rho_{\bm{z}'}  \hat{U}(\btheta')^{\dagger} )]_{j=1:d^N}\big\| \nonumber\\ 
=  &&\big\|[ \Tr(\hat{U}(\btheta')^{\dagger}\Pi_j \hat{U}(\btheta') (\rho_{\bm{z}}- \rho_{\bm{z}'}))]_{j=1:d^N}\big\| \nonumber\\
\leq && \Big\|[\big\| \rho_{\bm{z}}   - \rho_{\bm{z}'} \big\|]_{j=1:d^N}\Big\|  \nonumber\\
\leq && d^{N}  \frac{\epsilon}{3C_3 d^{N} \sqrt{n-1}} \nonumber\\
= && \frac{\epsilon}{3C_3\sqrt{n-1}},
\end{eqnarray}
 where the first inequality uses $\Tr(AB)\leq \Tr(A)\|B\|$ when $0\preceq A$ and $\Tr(\hat{U}(\btheta')^{\dagger}\Pi_j \hat{U}(\btheta'))=\Tr(\Pi_j)=1$ for $\forall j\in[d^N]$,  and the last inequality follows the definition of covering number. The achieved relation means that the covering set of $\mathcal{N}(\mathcal{Q}_3, \epsilon/(3C_3\sqrt{n-1}), \|\cdot \|_2)$ does not depend on $\btheta$ and its covering number is upper bounded by $\mathcal{N}(\mathcal{B}, \frac{\epsilon}{3C_3 d^{N} \sqrt{n-1}}, \|\cdot\|_2)$.

Then, based on the results in Lemma \ref{append:lem_cover_input_state},  we have
\begin{equation}\label{eqn:proof-lem4-3}
	\mathcal{N}\left(\mathcal{Q}_3, \frac{\epsilon}{3C_3\sqrt{n-1}}, \|\cdot \|_2\right)\leq  \mathcal{N}\left(\mathcal{B}, \frac{\epsilon}{3C_3 d^{N} \sqrt{n-1}}, \|\cdot\|_2\right)\leq \left(\frac{21C_3 d^{N}\sqrt{n-1} N_{ge} }{\epsilon} \right)^{d^{2k}N_{ge}}.
\end{equation}

\medskip
Combining Eqs.~(\ref{eqn:proof-lem4-1}) and (\ref{eqn:proof-lem4-3}), the covering number $ \mathcal{N}((\mathcal{G}_+)_{|\mathcal{S}}, \epsilon, \|\cdot \|_2)$ in Eqn.~(\ref{append:eqn:cov_g+_all}) is upper bounded by
\begin{eqnarray}
	 &&  \mathcal{N}\left((\mathcal{G}_+)_{|\mathcal{S}}, \epsilon, \|\cdot \|_2\right) \nonumber\\
	  \leq  && \mathcal{N}\left(\mathcal{H}_{circ}, \frac{\epsilon}{3C_3 2^{N-1}\sqrt{n-1} }, \|\cdot\|_2 \right) \times  \mathcal{N}\left(\mathcal{B}, \frac{\epsilon}{3C_3 2^{N-1} \sqrt{n-1}}, \|\cdot\|_2\right) \nonumber\\
	 \leq && \left(\frac{21C_3 d^{N}\sqrt{n-1} N_{gt}  \|\Pi_j\| }{\epsilon} \right)^{d^{2k}N_{gt}}\left(\frac{21C_3 d^{N}\sqrt{n-1} N_{ge} }{\epsilon} \right)^{d^{2k}N_{ge}} \nonumber\\
	 = && (21C_3 d^{N}\sqrt{n-1} N_{gt})^{d^{2k}N_{gt}}(21C_3 d^{N}\sqrt{n-1} N_{ge})^{d^{2k}N_{ge}}\left(\frac{1}{\epsilon}\right)^{d^{2k}(N_{ge}+N_{gt})}.
\end{eqnarray}

Denote $C_5 = (21C_3 d^{N}\sqrt{n-1} N_{gt})^{\frac{d^{2k}N_{gt}}{d^{2k}(N_{ge}+N_{gt})}}(21C_3 d^{N}\sqrt{n-1} N_{ge})^{\frac{d^{2k}N_{ge}}{d^{2k}(N_{ge}+N_{gt})}}$. Using Lemma \ref{lem:rad-cov-dudley}, we obtain 
\begin{eqnarray}\label{eqn:proof-lem4-0-2}
	\EE(T2)\leq && \frac{2}{\sqrt{n-1}} \inf_{\alpha>0} \left(4\alpha + \frac{12}{\sqrt{n-1}}\int_{\alpha}^{1}\sqrt{\ln\left(\left(\frac{C_5}{\epsilon}\right)^{d^{2k}(N_{ge}+N_{gt})}\right)} d\epsilon \right) \nonumber\\
	= && \frac{2}{\sqrt{n-1}} \inf_{\alpha>0} \left(4\alpha + \frac{12\sqrt{d^{2k}(N_{ge}+N_{gt})}}{\sqrt{n-1}}\int_{\alpha}^{1}\sqrt{\ln\left(\frac{C_5}{\epsilon}\right) } d\epsilon \right) \nonumber\\
	\leq && \frac{2}{\sqrt{n-1}} \inf_{\alpha>0} \left(4\alpha + \frac{12\sqrt{d^{2k}(N_{ge}+N_{gt})}}{\sqrt{n-1}}\left(\epsilon + \epsilon\ln\left(\frac{C_5}{\epsilon}\right) \right)\Big|_{\epsilon=\alpha}^1 \right).
\end{eqnarray}
 
For simplicity, we set $\alpha=1/\sqrt{n-1}$ in Eq.~(\ref{eqn:proof-lem4-0-2}) and then $\EE(T2)$  is upper bounded by
\begin{eqnarray}
	\EE(T2) && \leq \frac{2}{\sqrt{n-1}} \left( \frac{4}{\sqrt{n-1}} + \frac{12\sqrt{d^{2k}(N_{ge}+N_{gt})}}{\sqrt{n-1}}\left(\epsilon + \epsilon\ln\left(\frac{C_5}{\epsilon}\right) \right)\Big|_{\epsilon=\alpha}^1 \right) \nonumber\\
	&& \leq   \frac{8}{n-1} + \frac{24\sqrt{d^{2k}(N_{ge}+N_{gt})}}{n-1}\left(1+  \ln C_5 \right).   
\end{eqnarray}
Since the two exponent terms in $C_5$ are no larger than $1$, we have $C_5\leq (21C_3 d^{N}\sqrt{n-1} N_{gt})(21C_3 d^{N}\sqrt{n-1} N_{ge})$. This relation further simplifies the upper bound $\EE(T2)$ as
\begin{equation}
	\EE(T2) \leq \frac{8}{n-1} + \frac{24\sqrt{d^{2k}(N_{ge}+N_{gt})}}{n-1}\left(1+  N\ln (441dC_3^2(n-1)N_{ge}N_{gt}) \right).
\end{equation}
 
\end{proof}

\subsection{Proof of Lemma \ref{lem:cov-T3}}\label{appendix:lemcov-T3}
\begin{proof}[Proof of Lemma \ref{lem:cov-T3}]
	The proof of Lemma   \ref{lem:cov-T3} is very similar to the one of Lemma \ref{lem:cov-T2} and thus we skip it here.
\end{proof}

\section{ QGLMs in parameterized Hamiltonian learning}\label{append:strm}
Here we first explain how QGLMs advance GLMs in the task of parameterized Hamiltonian learning (PHL) from the theoretical view. Then we conduct numerical simulations to apply QGANs to tackle parameterized Hamiltonian learning problems.  Recall an $N$-qubit parameterized Hamiltonian is defined as $H(\bm{a})$, where $\bm{a}$ is the interaction parameter (e.g., the strength of the transverse magnetic field)  sampled from a prior  distribution $\mathbb{D}$ (e.g., uniform distribution). Denote $\ket{\phi(\bm{a})}$ as the ground state of $H(\bm{a})$.  PHL aims to use $m$ training samples $\{\bm{a}^{(i)}, \ket{\phi(\bm{a}^{(i)})}\}_{i=1}^m$ to approximate the distribution  of the ground states for $H(\bm{a})$ with $\bm{a}\sim \mathbb{D}$, i.e., $\ket{\phi(\bm{a})}\sim \QQ$. If a trained learning model can well approximate $\QQ$, then it can prepare the ground state of $H(\bm{a}')$ for an unseen parameter $\bm{a}'\sim \mathbb{D}$.

\subsection{Proof of Lemma \ref{lem:append:11}}
The proof of Lemma \ref{lem:append:11} is established on the results of quantum random circuits, which is widely believed to be classically computationally hard and in turn can be used to demonstrate quantum advantages on NISQ devices \cite{boixo2018characterizing,zhu2021quantum}.  The construction of a  random quantum circuit is as follows.  Denote  $\Hcal(N,s)$ as  the distribution over the quantum circuit $\Ccal$ under 2D-lattice $(\sqrt{N}\times \sqrt{N})$ structure, where $\Ccal$ is composed of $s$ two-qubit gates, each of them is drawn from the $2$-qubit Haar random distribution, and $s$ is required to be greater than the number of qubits $N$. For simplicity, here we choose $s = 2N^2$ to guarantee the hardness of simulating the distribution $\Hcal(N,s)$. The operating rule for the $i$-th quantum gate satisfies: 
 \begin{itemize}
     \item If $i\leq N$, the first qubit of the $i$-th gate is specified to be  the $i$-th qubit and the second is selected randomly from its neighbors; \;
     \item If $i> N$, the first qubit is randomly selected from $\{1,2,...,N\}$ and randomly select the second qubit from its neighbors.
 \end{itemize}
Following the same routine,  Ref.~\cite{aaronson2017complexity} proposed a Heavy Output Generation (HOG) problem detailed below to separate the power between classical and quantum computers on the distribution of the output quantum state after performing a circuit $\Ccal$ sampled from $\Hcal(N,s)$ on the initial state $\ket{0}^{\otimes N}$. 
\begin{definition}[HOG, \cite{aaronson2017complexity}]Given a random quantum circuit $\Ccal\sim \Hcal(N,s)$ for $s\geq N^2$, generate $k$ binary strings $z_1,z_2,\cdots, z_k$ in $\{0,1\}^N$ such that at least 2/3 fraction of $z_i$'s are greater than the median of all probabilities of the circuit outputs $\Ccal\ket{0}^{\otimes N}$.
 \end{definition}
Concisely, under the quantum threshold assumption, there do not exist classical samples that can spoof  the $k$ samples output by a random quantum circuit with success probability at least $0.99$ \cite{aaronson2017complexity}. In addition, they prove that   quantum computers can solve the HOG problem with high success probability. For completeness, we introduce the quantum threshold assumption as follows.
 
 \begin{assumption}
 [quantum threshold assumption,~\cite{aaronson2017complexity}] \label{assump:1}
 There is no polynomial-time classical algorithm that takes a random quantum circuit $\Ccal$ as input with $s\geq N^2$ and decides whether $0^n$ is greater than the median of all probabilities of $\Ccal$ with success probability $1/2 + \Omega\pbra{2^{-N}}$.
 \end{assumption}

We are now ready to prove Lemma \ref{lem:append:11}. 
\begin{proof}[Proof of Lemma \ref{lem:append:11}]
The core idea of the proof is to show that there exists a ground state $\ket{\phi(\bm a)}$ of a Hamiltonian $H(\bm a)$, which can be efficiently prepared by quantum computers but is computationally hard for classical algorithms. To achieve this goal, we connect $\ket{\phi(\bm a)}$  with the output state of random quantum circuits.

With the quantum threshold assumption in Assumption \ref{assump:1}, Aarsonson and Chen \cite{aaronson2017complexity} proved that there exists a quantum state $\ket{\Phi}$ generated from a random circuit $\Ccal$ sampling from $\Hcal\pbra{N, 2N^2}$ such that HOG problem on this instance is  classically hard,  with success probability at least $0.99$. Conversely, $\ket{\Phi}$ can be efficiently prepared by the parameterized quantum circuit $\hat{U}(\btheta^*)$ whose topology is identical to $\Ccal$. Moreover, due to the fact that quantum adiabatic algorithms with 2-local Hamiltonians can implement universal quantum computational tasks \cite{kempe2006complexity}, the quantum state $\ket{\Phi}\equiv \ket{\phi(\bm{a}^*)}= \hat{U}(\btheta^*)\ket{0}^{\otimes N}$ must correspond to the ground state of a certain Hamiltonian $H(\bm{a}^*)$. 

We now leverage the above result to design a parameterized Hamiltonian learning task that separates the power of classical and quantum machines. In particular, we restrict the target distribution $\QQ$, or equivalently $\mathbb{D}$, as the delta distribution, where the probability of sampling $\ket{\Phi}\equiv \ket{\phi(\bm{a}^*)}=\hat{U}(\btheta^*)\ket{0}^{\otimes N}$ equals to one and the probability of sampling other ground states of $H(\bm{a}')$ with $\bm{a}'\neq \bm{a}^*$ is zero. In this way, the Hamiltonian learning task is reduced to using QGLM or GLM to prepare the quantum state $\ket{\Phi}$. This task can be efficiently achieved by QGML but is computationally hard for GLMs.
 \end{proof}

\subsection{Numerical simulation details}
We apply QGANs introduced in the main text to study the parameterized Hamiltonian learning problem. In particular, the parameterized Hamiltonian is specified as the XXZ spin chain, i.e.,
\begin{equation}
	H(\bm{a}) = \sum_{i=1}^N (X_iX_{i+1}+Y_iY_{i+1} + \bm{a}Z_iZ_{i+1})+\eta\sum_{i=1}^N Z_i. 
\end{equation}   
In all numerical simulations, we set $N=2$ and $\eta=0.25$. The distribution $\mathbb{D}$ for the parameter $\bm{a}$ is uniform ranging from $-0.2$ to $0.2$. In the preprocessing stage, to collect $m$ referenced samples $\{\bm{a}^{(j)}, \ket{\phi(\bm{a}^{(j)})}\}_{j=1}^m$, we first uniformly sample $\{\bm{a}^{(j)}\}_{j=1}^m$ points from $\mathbb{D}$ and calculate the corresponding eigenstates of $\{H(\bm{a}^{(j)})\}_{j=1}^m$ via the exact diagonalization. 

\begin{figure*}
	\centering
\includegraphics[width=0.95\textwidth]{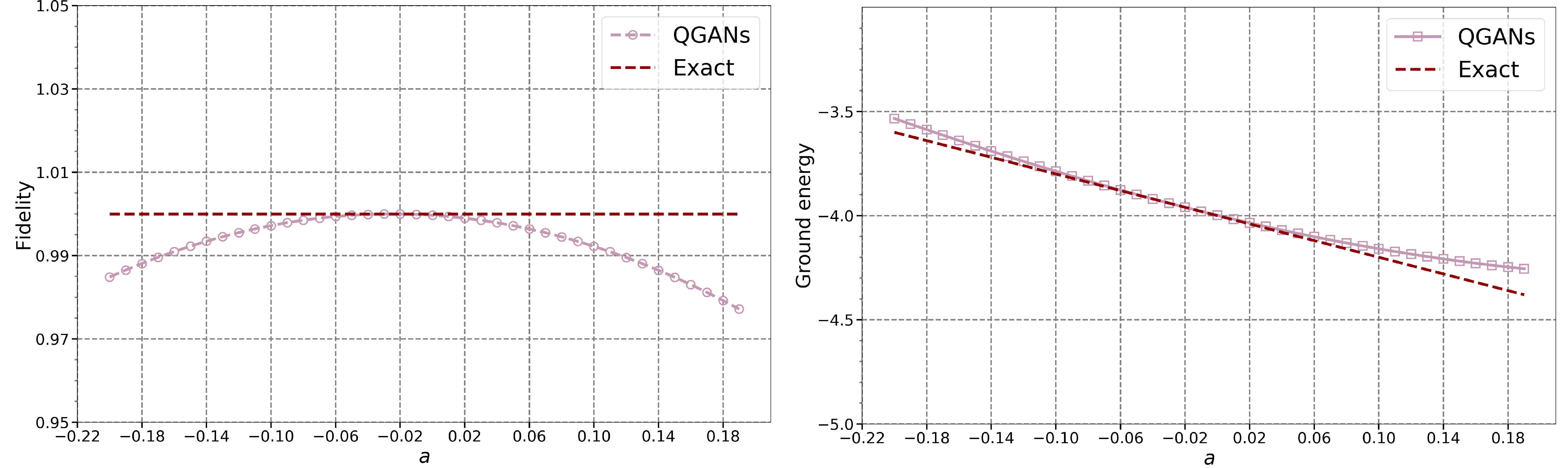}
\caption{\small{\textbf{Simulation results of QGANs in parameterized Hamiltonian learning.} The left panel shows the fidelity of the approximated ground states and the exact ground states of Hamiltonian $H(\bm{a})$ with varied $\bm{a}$. The right panel exhibits the estimated and exact ground energies of a class of Hamiltonians $H(\bm{a})$. }}
	\label{fig:append:phl}
\end{figure*}

The setup of QGAN is as follows. The prior distribution $\PP_{\mathcal{Z}}$ is set as $\mathbb{D}$. The encoding unitary is $U(\bm{a})=\text{CNOT}(\text{RZ}(\bm{a})\text{RY}(\bm{a}))\otimes (\text{RZ}(\bm{a})\text{RY}(\bm{a}))$. The hardware-efficient Ansatz is used to implement $U(\bm{\theta})=\prod_{l=1}^L U_l(\bm{\theta})$ with $U_l(\bm{\theta})=\text{CNOT}(\text{RZ}(\bm{\theta}_{l,1})\text{RY}(\bm{\theta}_{l,2}))\otimes (\text{RZ}(\bm{\theta}_{l,3})\text{RY}(\bm{\theta}_{l,4}))$. The number of blocks is $L=4$. The number of training and referenced examples is set as $n=m=9$. The Adagrad optimizer is used to update $\btheta$. The total number of iterations is set as $T=80$. We employ $5$ different random seeds to collect the statistical results. To evaluate the performance of the trained QGAN, we apply it to generate in total $41$ ground states of $H(\bm{a})$ with $\bm{a}\in [-0.2, 0.2]$ and compute the fidelity with the exact ground states.
 
The simulation results are shown in Fig.~\ref{fig:append:phl}. Specifically, for all settings of $\bm{a}$,  the fidelity between the approximated ground state output by QGANs and the exact ground state is above 0.97. Moreover, when $\bm{a}\in [-0.06, 0.06]$, the fidelity is near to $1$. The right panel depicts the estimated ground energy using the output of QGANs, where the maximum estimation error is $0.125$ when $\bm{a}=0.2$. These observations verify the ability of QGANs in estimating the ground states of parameterized Hamiltonians.

\section{More details of numerical simulations}\label{appendix:sim}

\subsection{Hyper-para and metrics}

\textbf{RBF kernel.}  The explicit expression of the radial basis function (RBF) kernel is $k(\rx, \ry) = \exp (-\frac{||\rx-\ry||^2}{\sigma^2})$, where $\sigma$ refers to the bandwidth. In all simulations, we set $ \sigma^{-2}$ as $\{0.25, 4\}$ for QCBMs and $\{-0.001,1, 10\}$ for QGANs, respectively.

\textbf{KL divergence}. We use the KL divergence to measure the similarity between the generated distribution $\PP_{\btheta}$ and true distribution $\QQ$. In the discrete setting, its mathematical expression is $\text{KL}(\PP_{\btheta}||\QQ)=\sum_i \PP_{\btheta}(i)\log(\QQ(i)/\PP_{\btheta}(i))\in [0, \infty)$. In the continuous setting, $\text{KL}(\PP_{\btheta}||\QQ)=\int \PP_{\btheta}(d\rx)\log(\QQ(d\rx)/\PP_{\btheta}(d\rx)) d\rx \in [0, \infty)$. When the two distributions are exactly matched with $\PP=\QQ$, we have $\text{KL}(\PP||\QQ)=0$.

\textbf{State fidelity.} Suppose that the generated state is $\ket{\Psi(\btheta)}$ and the target state is $\ket{\Psi^*}$. The state fidelity \cite{nielsen2010quantum} for pure states is defined as $F=|\braket{\Psi(\btheta)|\Psi^*}|^2$. 

\textbf{Optimizer.} For QCBMs, the classical optimizer is assigned as L-BFGS-B algorithm and the  tolerance for termination is set as $10^{-12}$. For QGANs, the classical optimizer is assigned as Adam with default parameters.

\textbf{Hardware parameters}. All simulation results in this study are completed by the classical device with Intel(R) Xeon(R) Gold 6267C CPU  @ 2.60GHz and 128 GB memory.

 \textbf{Data and code availability} 
 The source code for conducting all numerical simulations will be available in Github repository \url{https://github.com/yuxuan-du/QGLM-Theory}.

\subsection{Simulation results related to the task of discrete Gaussian approximation}

\textbf{Training loss of QCBMs.}  Fig.~\ref{fig:append_sim_QCBM_gau}(a) plots the last iteration  training loss of QCBMs. All hyper-parameter settings are identical to those introduced in the main text. The $x$-axis stands for the setting of $n$ used to compute $\MMD_U$ in Eq.~(\ref{append:eqn:unbia-MMD}). The simulation results indicate that the performance QCBM with RBM kernel is steadily enhanced with the increased $n$. When $n\rightarrow \infty$, its performance approaches to the QCBM with quantum kernel. These phenomenons accord with Theorem \ref{thm:Gene-QCBM}.

\begin{figure*}
	\centering
\includegraphics[width=0.98\textwidth]{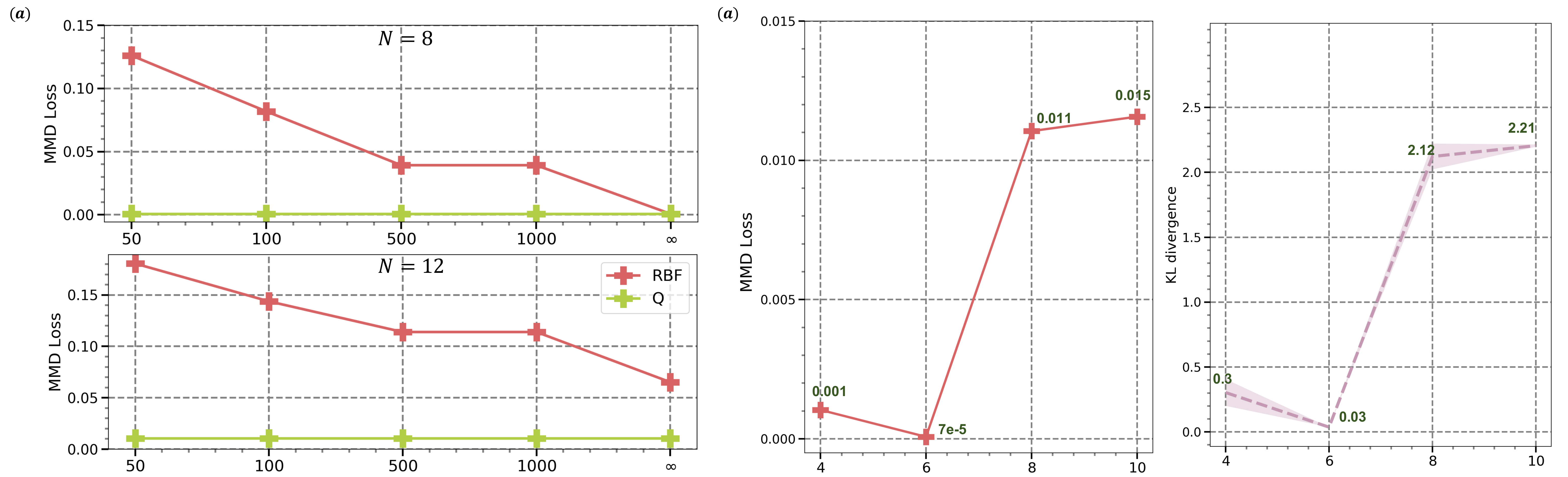}
	\caption{\small{\textbf{Generating discrete Gaussian with varied settings of QCBMs.} (a) The simulation results with the varied number samples $n$ (corresponding to $x$-axis). The upper and lower panels demonstrate the achieved MMD loss for $N=8, 12$, respectively. (b) The simulation results of QCBM  with $N=12$, the quantum kernel, and  $n\rightarrow\infty$ for the varied circuit depth   (corresponding to $x$-axis). The left and right panels separately show the achieved MMD loss and the KL divergence between the generated and the target distributions. }}
	\label{fig:append_sim_QCBM_gau}
\end{figure*}

\textbf{Effect of circuit depth.}
We explore the performance of QCBMs with quantum kernels by varying the employed Ansatz. Specifically, we consider the case of $N=12$ and set $L_1$ in Fig.~\ref{fig:dist_main_QCBM}(a) as $4, 6, 8, 10$. The collected simulation results are shown in Fig.~\ref{fig:append_sim_QCBM_gau}(b). In conduction with Fig.~\ref{fig:dist_main_QCBM}(c), QCBM with $L_1=6$ attains the best performance over all settings, where the achieved MMD loss is $7\times 10^{-5}$ and the KL divergence is $0.03$. This observation implies that properly controlling the expressivity of Ansatz, which effects the term $C_1$ in Theorem \ref{thm:Gene-QCBM}, contributes to improve the learning performance of QCBM.

\subsection{More simulation results related to the task of GHZ state approximation}

 The approximated GHZ states of QGANs with different random seeds discussed in Fig.~\ref{fig:dist_main_QCBM} are depicted in Fig.~\ref{fig:tomo-ghz}. Specifically, the difference between the approximated state and the target GHZ state becomes apparent with the decreased number of examples $n$ and the increased number of qubits $N$. These observations echo with the statement of Theorem \ref{thm:Gene-QCBM}. 

\begin{figure*}[!htb]
	\centering
\includegraphics[width=0.98\textwidth]{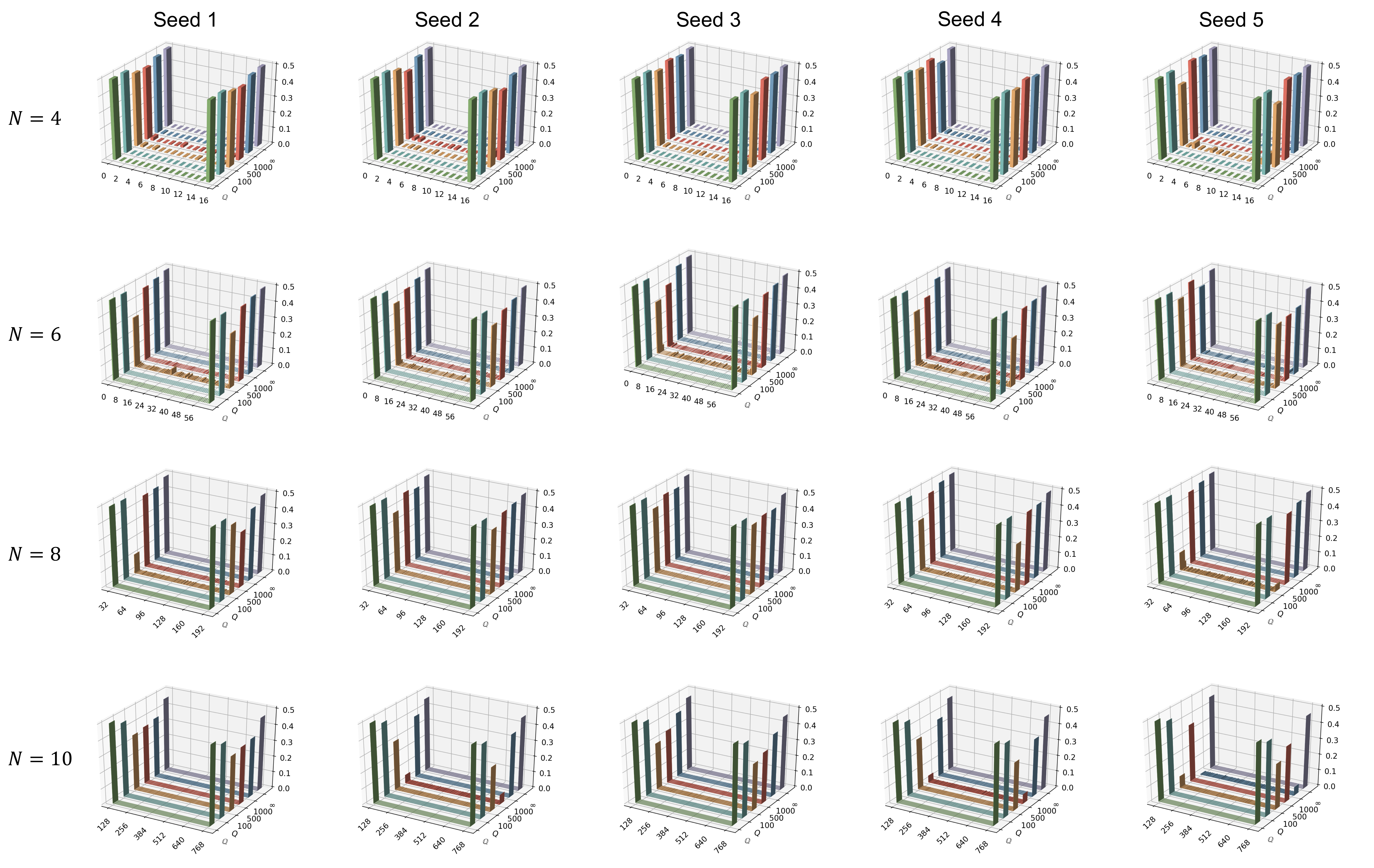}
\caption{\small{\textbf{The approximated GHZ state with the varied number of qubits.}  The label follows the same meaning explained in Fig.~\ref{fig:dist_main_QCBM}.  }}
\label{fig:tomo-ghz}
\end{figure*}

\subsection{More simulation results related to the task of 3D Gaussian approximation}

\subsubsection{Implementation of the modified style-QGAN} 
The construction details  of the modified style-QGAN, especially for  $U(\bm{z})$ and $U(\btheta)$, are illustrated in Fig.~\ref{fig:style-QGAN}. Particularly, the circuit layout of  $U(\bm{z})$ and the $l$-th layer of $\hat{U}(\btheta)$ is the same. Mathematically,  $U(\bm{z})=U_E(\bm{\gamma}_4)(\otimes_{i=1}^3  U(\bm{\gamma}_i))$, where $U(\bm{\gamma}_i))=\RZ(\bm{z}_{3})\RY(\bm{z}_{2})\RZ(\bm{z}_{2})\RY(\bm{z}_{1}), \forall i \in [3]$ and $U_E(\bm{\gamma}_4)=(\mathbb{I}_2 \otimes \CRY(\bm{z}_2))(\CRY(\bm{z}_1)\otimes \mathbb{I}_2)$ refers to the entanglement layer. Similarly, for the $l$-th layer of $\hat{U}(\btheta)$, its mathematical expression is $\hat{U}_l(\btheta)=U_E(\bm{\gamma}_4)(\otimes_{i=1}^3  U(\bm{\gamma}_i))$, where $U(\bm{\gamma}_i))=\RZ(\btheta_{3})\RY(\btheta_{2})\RZ(\btheta_{2})\RY(\btheta_{1}), \forall i \in [3]$. When $l$ is odd, the entanglement layer takes the form $U_E(\bm{\gamma}_4)=(\mathbb{I}_2 \otimes \CRY(\btheta_2))(\CRY(\btheta_1)\otimes \mathbb{I}_2)$; otherwise, its implementation is shown in the lower right panel of Fig.~\ref{fig:style-QGAN}.

\begin{figure}[h!]
	\centering
\includegraphics[width=0.48\textwidth]{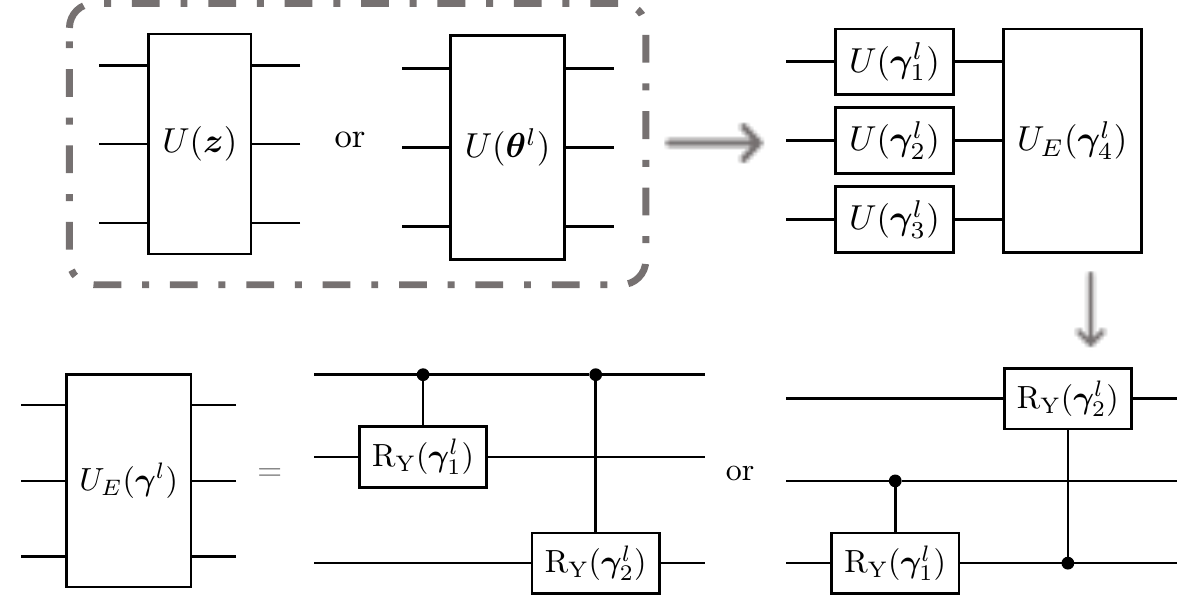}
\caption{\small{\textbf{The implementation detail of the modified style-QGANs.}     The implementation of QGANs when the number of qubits is $N=3$. $U(\bm{z})$ refers to the encoding circuit to load the example $\bm{z}$. The meaning of `$L_1-1$' is identical to the one explained in Fig.~\ref{fig:3D-gau}. The gate $U(\btheta_i^l)$ refers to the $\RY\RZ$ gates applied on the $i$-th qubit in the $l$-th  layer of $\hat{U}(\btheta)$.   The circuit architecture of  $U(\bm{z})$ and the $l$-th layer of $\hat{U}(\btheta)$ is identical, which is depicted in the upper right panel. The lower panel plots the construction of the entanglement layer $U_E(\bm{\gamma^l})$. }}
\label{fig:style-QGAN}
\end{figure}

  The optimization of the modified style-QGAN follows an iterative manner. At each iteration, a classical optimizer leverages the batch gradient descent method  to update the trainable parameters $\btheta$ minimizing  MMD loss. After $T$ iterations, the optimized parameters are output as the estimation of the optimal results. The Pseudo code of the modified style-QGAN is summarized in Alg.~\ref{alg:QGAN}.

\begin{algorithm}[H]
  \SetAlgoLined
  \KwData{Training set $\{\ry^{i}\}_{i=1}^m$,    number of examples $n$,  learning rate $\eta$,  iterations $T$, MMD loss;}
  \KwResult{Output the optimized parameters.}
Randomly divide   $\{\ry^{(i)}\}_{i=1}^m$ into $m_{\text{mini}}$ mini batches with batch size $b$\;
  Initialize parameters $\btheta$\;
  \While{$T>0$}{
    Regenerate noise inputs $\{\bm{z}^{(i)}\}_{i=1}^n$ every $r$ iterations\;
    \For{$j\leftarrow 1, m_{\text{mini}}$}{
    Generate $\{\rx^{(i)}\}_{i=1}^n$ with $\rx^{(i)}=G_{\btheta}(\bm{z}^{(i)})$\;
    Compute the $b$'th minibatch’s gradient $\nabla \MMD_U^2(\PP_{\btheta}^{n}||\QQ^{b})$ \;
     $\btheta \gets \btheta  - \eta \nabla \MMD_U^2(\PP_{\btheta}^{n}||\QQ^{b}) $ \;
    }
    $T\leftarrow T -1$\; 
    }
  \caption{The modified style-QGAN}
  \label{alg:QGAN}
\end{algorithm}

\subsubsection{More simulation results}
We next examine how the number of examples $m$ and the number of trainable gates $N_g$ effect the generalization of QGANs. The experimental setup is identical to those introduced in the main text. To attain a varied number of $N_g$, the circuit depth of Ansatz in Fig.~\ref{fig:3D-gau}(a) is set as $L=2, 4, 6, 8$. Other hyper-parameters are fixed with $T=800$, $n=m=5000$, batch size $b=64$. We repeat each setting with $5$ times to collect the simulation results. 

 \textbf{Effect of the number of examples.} Let us first focus on the setting $L=2$ and $m=5000$. In conjunction with the simulation results in the main text (i.e., $L=2$ and $m=2, 10, 200$), the simulation results in Fig.~\ref{fig:append-sim-2dgau} indicate that an increased number of $n$ and $m$ contribute to a better generalization property. Specifically, at $t=120$, the averaged  empirical $\MMD$ loss ($\MMD_U$) of QGANs is $0.0086$, which is comparable with other settings discussed in the main text. The averaged expected $\MMD$ loss is $0.0045$, which is similar to the setting with $m=200$. In other words, when the number of examples $m$ and $n$ exceeds a certain threshold, the generalization error of QGANs is dominated by other factors instead of $m$ and $n$.

\textbf{Effect of the number of trainable gates.} We next study how the number of trainable gates effects the generalization error of QGANs. Following the structure of the employed Ansatz shown in Fig.~\ref{fig:3D-gau}(a), varying the number of trainable gates amounts to varying the number of blocks $L_1$. The results of QGANs with varied $L_1$ are illustrated in Fig.~\ref{fig:append-sim-2dgau}. For all setting of QGANs, their empirical $\MMD$ loss fast converges after $40$ iterations. Nevertheless, their  expected $\MMD$ loss is distinct, where a larger $L_1$ (or equivalently, a larger number of trainable gates $N_g$) implies a higher expected $\MMD$ loss and leads to a worse generalization. These observations accord with the result of Theorem \ref{thm:Gene-QGAN} in the sense that an Ansatz with the overwhelming expressivity may incur a poor generalization ability of QGANs.

\begin{figure*}
	\centering
\includegraphics[width=0.98\textwidth]{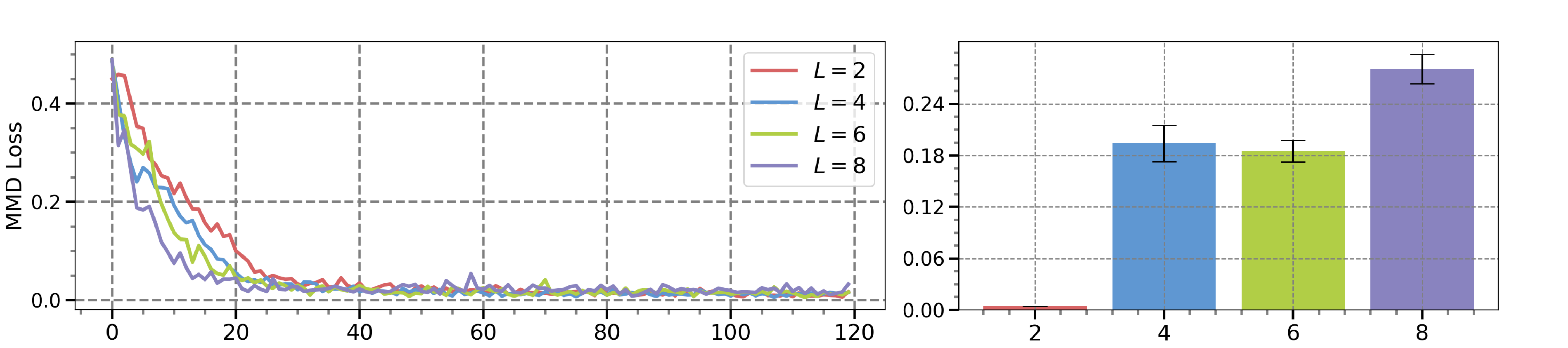}
\caption{\small{\textbf{The simulation results of QGANs with the varied number of quantum gates}. The left panel shows the MMD loss of QGANs during $120$ iterations. The label `$L=a$' refers to set the block number as $L_1=a+1$. The right panel evaluates the generalization property of trained QGANs by calculating the expected MMD loss. The x-axis refers to $L$.} }
\label{fig:append-sim-2dgau}
\end{figure*}

\end{document}